%% file: arxiv.tex
\documentclass[letter]{article}

\usepackage[utf8]{inputenc} % allow utf-8 input
\usepackage[T1]{fontenc}    % use 8-bit T1 fonts
\usepackage{hyperref}       % hyperlinks
\usepackage{url}            % simple URL typesetting
\usepackage{booktabs}       % professional-quality tables
\usepackage{amsfonts}       % blackboard math symbols
\usepackage{nicefrac}       % compact symbols for 1/2, etc.
\usepackage{microtype}      % microtypography
\usepackage{graphicx}
\usepackage{subfigure}
\usepackage{amsthm}    
\usepackage{graphicx}
\usepackage{subfigure}
\usepackage{multirow}
\usepackage{hhline}
\usepackage{algorithm}
\usepackage{enumerate}
\usepackage{amsmath}
\usepackage{algorithmic}
\usepackage{mathtools}
\usepackage{hyperref}
\usepackage{wrapfig}
\usepackage{xspace}

\newif\ifappendix
\appendixtrue

\newcommand{\refappendix}[1]{\ifappendix
 Appendix~\ref{#1}\xspace
\else
 the supplementary material\xspace
\fi}

\usepackage[textsize=small]{todonotes}

\usepackage{makecell}

\DeclareMathOperator{\dist}{d}
\DeclareMathOperator{\opt}{OPT}
\DeclareMathOperator{\activewindow}{\mathrm{W}}

\DeclareMathOperator{\argmin}{argmin}

\DeclareMathOperator{\weight}{\mathsf{weight}}

\DeclareMathOperator{\last}{\mathsf{last}}

\mathchardef\mhyphen="2D

\renewcommand{\O}{\mathcal{O}}

\newcommand{\cost}{f}
\newcommand{\acost}{\widehat{\cost}}
\newcommand{\costm}{\mathsf{cost_\mu}}

\newcommand{\centers}{\mathcal{C}}

\newcommand{\cen}{\mathsf{c}}
\newcommand{\suffix}{\mathsf{Suffix}}
\newcommand{\calC}{\mathcal{C}}
\newcommand{\cl}{\mathsf{cl}}

\newcommand{\algo}{\mathsf{ALG}}

\newcommand{\comment}[1]{}

\newtheorem{theorem}{Theorem}[section]

\newtheorem{lemma}[theorem]{Lemma}
\newtheorem{definition}[theorem]{Definition}

\title{Sliding Window Algorithms for k-Clustering Problems}

\author{%
  Michele Borassi, Alessandro Epasto, Silvio Lattanzi,\\
  Sergei Vassilviskii, Morteza Zadimoghaddam\\
  Google \\
  \texttt{\{borassi,aepasto,silviol,sergeiv,zadim\}@google.com}\\
}

\begin{document}

\maketitle

\begin{abstract}
The sliding window model of computation captures scenarios in which data is arriving continuously, but only the latest $w$ elements should be used for analysis. The goal is to design algorithms that update the solution efficiently with each arrival rather than recomputing it from scratch. In this work\footnote{This manuscript is an extended version of the paper ``Sliding Window Algorithms for k-Clustering Problems'' to appear in Neurips2020.}, we focus on $k$-clustering problems such as $k$-means and $k$-median. In this setting, we provide simple and practical algorithms that offer stronger performance guarantees than previous results. Empirically, we show that our methods store only a small fraction of the data, are orders of magnitude faster, and find solutions with costs only slightly higher than those returned by algorithms with access to the full dataset. 
\end{abstract}

\input{intro}

\input{prelim}

\section{Algorithm and Analysis}
\label{sec:algorithm}

The starting point of our clustering is the development of efficient sketching technique that, given a stream of points, $X$, a mapping $\mu$, and a time, $\tau$,  returns a weighted instance that is $\epsilon$-consistent with $\mu$ for the points inserted at or after $\tau$. To see why having such a sketch is useful, suppose $\mu$ has a cost a constant factor larger than the cost of the optimal solution. Then we could get an approximation to the sliding window problem by computing an approximately optimal clustering on the weighted instance (see Lemma~\ref{lem:cmap}). 

To develop such a sketch, we begin by relaxing our goal by allowing our sketch to return a weighted instance that is $\epsilon$-consistent with $\mu$ for the {\it entire stream} $X$ as opposed to the substream starting at $X_{\tau}$. Although a single sketch with this property is not enough to obtain a good algorithm for the overall problem, we design a sliding window algorithm that builds multiple such sketches in parallel. We can show that it is enough to maintain a polylogarithmic number of carefully chosen sketches to guarantee that we can return a good approximation to the optimal solution in the active window.

In subsection~\ref{sec:meyerson} we describe how we construct a single efficient sketch. Then, in the subsection~\ref{sec:algo}, we describe how we can combine different sketches to obtain a good approximation. 
All of the missing proofs of the lemmas and the pseudo-code for all the missing algorithms are presented \ifappendix 
 in the appendix.
\else in the supplementary material.
\fi

\subsection{Augmented Meyerson Sketch}
\label{sec:meyerson}

Our sketching technique builds upon previous clustering algorithms developed for the streaming model of computation. Among these, a powerful approach is the sketch introduced for facility location problems by Meyerson~\cite{meyerson2001online}.

At its core, given an approximate lower bound to the value of the optimum solution, Meyerson's algorithm constructs a set $\centers$ of size $O(k \log \Delta)$, known as a {\em sketch}, and a consistent weighted instance, such that, with constant probability, $\cost_p(X, \centers) \in O(\opt_p(X))$. 
Given such a sketch, it is easy to both: amplify the success probability to be arbitrarily close to 1 by running multiple copies in parallel, and reduce the number of centers to $k$ by keeping track of the number of points assigned to each $\cen \in \centers$ and then clustering this weighted instance into $k$ groups. 

What makes the sketch appealing in practice is its easy construction---each arriving point is added as a new center with some carefully chosen probability. If a new point does not make it as a center, it is assigned to the nearest existing center, and the latter's weight is incremented by 1. 

Meyerson algorithm was initially designed for online problems, and then adapted to algorithms in the streaming computation model, where points arrive one at a time but are never deleted. To solve the sliding window problem naively, one can simply start a new sketch with every newly arriving point, but this is inefficient. To overcome these limitations we extend the Meyerson sketch. In particular, there are two challenges that we face in sliding window models: 
\begin{enumerate}\itemsep=0in
\item The weight of each cluster is not monotonically increasing, as points that are assigned to the cluster time out and are dropped from the window. 
\item The designated center of each cluster may itself expire and be removed from the window, requiring us to pick a new representative for the cluster. 
\end{enumerate}

Using some auxiliary bookkeeping we can augment the classic Meyerson sketch to return  a weighted instance that is $\epsilon$-consistent with a mapping $\mu$ whose cost is a constant factor larger than the cost of the optimal solution for the entire stream $X$. More precisely, 

\begin{lemma}\label{lemma:mb}
Let $w$ be the size of the sliding window, $\epsilon\in (0,1)$ be a constant and $t$ the current time. Let $(X,\dist)$ be a metric space and fix $ \gamma \in (0,1)$. The augmented Meyerson algorithm computes
an \emph{implicit} mapping $\mu: X \rightarrow \centers$, and
 an $\epsilon$-consistent weighted instance $(\centers, \widehat{\weight})$ for all substreams $X_{[\tau, t]}$ with $\tau \geq t - w$, such that,  with probability $1- \gamma$, we have: 
$ |\centers| \leq 2^{2p+8} k \log  \gamma^{-1}\log \Delta \qquad {\text{and}} \qquad $
$ \cost_p(X_{[\tau,t]}, \centers)  \le  2^{2p+8} \opt_p(X).$

The algorithm uses space $O( k \log  \gamma^{-1} \log \Delta \log (M/m) (\log M + \log w + \log \Delta))$ and stores the cost of the consistent mapping, $\cost(X, \mu)$, and allows a $1+\epsilon$ approximation to the cost of the $\epsilon$-consistent mapping, denoted by $\widehat{\cost}(X_{[\tau, t]}, \mu)$.
This is the $\epsilon$-consistent mapping that is computed by the augmented Meyerson algorithm.
In section~\ref{sec:preliminaries}, $M$ and $m$ are defined as the upper and lower bounds on the cost of the optimal solution. 
\end{lemma}

Note that when $M/m$ and $\Delta$ are polynomial in $w$,\footnote{We note that prior work~\cite{DBLP:conf/soda/BravermanLLM16,cohen2016diameter} makes similar assumptions to get a bound depending on $w$.} the above space bound is $O(k \log \gamma^{-1} \log^3(w)$).

\input{algo}

\input{experiments}

\section{Conclusion}
We present the first  algorithms for the $k$-clustering problem on sliding windows with  space linear in $k$. Empirically we observe that the algorithm performs much better than the analytic bounds, and it allows to store only a small fraction of the input. A natural avenue for future work is to give a tighter analysis, and reduce this gap between theory and practice. 

\section*{Broader Impact}
Clustering is a fundamental unsupervised machine learning problem that lies at the core of multiple real-world applications. In this paper, we address the problem of clustering in a sliding window setting. As we argued in the introduction, the sliding window model allows us to discard old data which is a core principle in data retention policies.

Whenever a clustering algorithm is used on user data it is  important to consider the impact it may have on the users. In this work we focus on the algorithmic aspects of the problem and we do not address other considerations of using clustering that may be needed in practical settings. For instance, there is a burgeoning literature on fairness considerations in unsupervised methods, including clustering, which further delves into these issues. We refer to this literature~\cite{chierichetti2017fair,kearns2019ethical,barocas-hardt-narayanan} for addressing such issues.
\bibliographystyle{plain}
\bibliography{references}

\ifappendix
\newpage
\onecolumn

\appendix
\input{app-known}

\input{app-aug}

\input{app-sw}
\input{app-exp}
\input{app-assumption}
\fi

\end{document}

%% file: intro.tex
\section{Introduction}
\label{sec:intro}

Data clustering is a central tenet of unsupervised machine learning. One version of the problem can be phrased as grouping data into $k$ clusters so that elements within the same cluster are similar to each other. Classic formulations of this question include the $k$-median and $k$-means problems for which good approximation algorithms are known~\cite{DBLP:conf/focs/AhmadianNSW17,DBLP:journals/siamcomp/LiS16}. Unfortunately, these algorithms often do not scale to  large modern datasets requiring researchers to turn to parallel~\cite{kmeansparallel},  distributed~\cite{distributedkmeans}, and streaming methods. In the latter model, points arrive one at a time and the goal is to find algorithms that quickly update a small {\em sketch} (or summary) of the input data that can then be used to compute an approximately optimal solution. 

One significant limitation of the classic data stream model is that it ignores the time when a data point arrived; in fact, all of the points in the input are treated with equal significance. However, in practice, it is often important (and sometimes necessary) to restrict the computation to very recent data. This restriction may be due to data freshness---e.g., when training a model on recent events, data from many days ago may be less relevant compared to data from the previous hour. Another motivation arises from legal reasons, e.g., data privacy laws such as the General Data Protection Regulation (GDPR), encourage and mandate that companies not retain certain user data beyond a specified period. This has resulted in many products including a data retention policy~\cite{upadhyay2019sublinear}. Such recency requirements can be modeled by the {\em sliding window} model. Here the goal is to maintain a small sketch of the input data, just as with the streaming model, and then use only this sketch to approximate the solution on the last $w$ elements of the stream. 

Clustering in the sliding window model is the main question that we study in this work. A trivial solution simply maintains the $w$ elements in the window and recomputes the clusters from scratch at each step. We intend to find solutions that use less space, and are more efficient at processing each new element.  In particular, we present an algorithm which uses space linear in $k$, and polylogarithmic in $w$, but still attains a constant factor approximation.

\paragraph{Related Work}
\label{sec:related}
{\noindent \emph{Clustering.}}
Clustering is a fundamental problem in unsupervised machine learning and has application in a disparate variety of settings, including data summarization, exploratory data analysis,  matrix approximations and outlier detection~\cite{jain2010data,pmlr-v97-kleindessner19a,malkomes2015fast,oglic2017nystrom}. One of the most studied formulations in clustering of metric spaces is that of finding $k$ centers that minimize an objective consisting of the $\ell_p$ norm of the distances of all points to their closest center. For $p\in\{1,2,\infty\}$ this problem corresponds to $k$-median, $k$-means, and $k$-center, respectively, which are NP-hard, but constant factor approximation algorithms are known~\cite{DBLP:conf/focs/AhmadianNSW17,gonzalez1985clustering,DBLP:journals/siamcomp/LiS16}. Several techniques have been used to tackle these problems at scale, including dimensionality reduction~\cite{makarychev2019performance}, core-sets~\cite{DBLP:conf/aistats/BachemLL18}, distributed algorithms~\cite{bachem2017distributed}, and streaming methods reviewed later. To clarify between Euclidean or general metric spaces, we note that our results work on arbitrary general metric spaces. The hardness results in the literature hold even for special case of Euclidean metrics and the constant factor approximation algorithms hold for the general metric spaces.

{\noindent \emph{Streaming model.}}
Significant attention has been devoted to models for analyzing large-scale datasets that evolve over time. The streaming model of computation is of the most well-known (see~\cite{streams} for a survey) and focuses on defining low-memory algorithms for processing data arriving one item at a time. A number of interesting results are known in this model ranging from the estimation of stream statistics~\cite{AMS, UniqueStream}, to submodular optimization~\cite{KDD14}, to graph problems~\cite{l0sampling, semistreaming, krauthgamer2019almost}, and many others. Clustering is also well studied in this setting, including algorithms for $k$-median, $k$-means, and $k$-center in the insertion-only stream case~\cite{DBLP:conf/aistats/BachemLL18,charikar03better, Guha:2000:CDS:795666.796588}.  

{\noindent \emph{Sliding window streaming model.}}
The sliding window model significantly increases the difficultly of the problem, since deletions need to be handled as well. Several techniques are known, including the exponential histogram framework~\cite{datar2002maintaining} that addresses weakly additive set functions, and the smooth histogram framework~\cite{Braverman07} that is suited for functions that are well-behaved and possesses a sufficiently small constant approximation. Since many  problems, such as $k$-clustering, do not fit into these two categories, a number of algorithms have been developed for specific problems such as submodular optimization~\cite{SlidingDiversityPODS2019,ChenNZ16,epasto2017submodular}, graph sparsification~\cite{crouch2013dynamic}, minimizing the enclosing ball~\cite{wang2019coresets}, private heavy hitters~\cite{upadhyay2019sublinear}, diversity maximization~\cite{SlidingDiversityPODS2019} and linear algebra operations~\cite{braverman2018near}. Sliding window algorithms find also applications in data summarization~\cite{chrysos2019}.

Turning to sliding window algorithms for clustering, for the $k$-center problem Cohen et al.~\cite{cohen2016diameter} show a $(6+\epsilon)$-approximation using $O(k \log \Delta)$ space and per point update time of $O(k^2 \log \Delta)$, where $\Delta$ is the spread of the metric, i.e. the ratio of the largest to the smallest pairwise distances.  
For $k$-median and $k$-means,~\cite{DBLP:conf/soda/BravermanLLM16} give constant factor approximation algorithms that use $O(k^3 \log^6 w)$ space and per point update time of $O(poly(k, \log w))$.\footnote{We note that the authors assume that the cost of any solution is polynomial in $w$. We chose to state our bounds explicitly, which introduces a dependence on the ratio of the max and min costs of the solution.}

Their bound is polylogarithmic in $w$, but {\em cubic} in $k$, making it impractical unless $k \ll w$.\footnote{We note here that in some practical applications $k$ can be large. For instance, in spam and abuse~\cite{sheikhalishahi2015fast}, near-duplicate detection~\cite{hassanzadeh2009framework} or reconciliation tasks~\cite{rostami2018interactive}.} In this paper we improve their bounds and  give a simpler algorithm with only linear dependency of $k$. Furthermore we show experimentally (Figure~\ref{fig:soda} and Table~\ref{table:soda16}) that our algorithm is faster and uses significantly less memory than the one presented in~\cite{DBLP:conf/soda/BravermanLLM16} even with very small values $k$ (i.e.,  $k\geq 4$). In a different approach, \cite{youn2018efficient} study a variant where one receives points in batches and uses heuristics to reduce the space and time. Their approach does provide approximation guarantees but it applies only to the Euclidean k-means case. Recently, ~\cite{gayen2018new} studied clustering problems in the distributed sliding window model, but these results are not applicable to our setting.

The more challenging fully-dynamic stream case has also received attention~\cite{braverman2017clustering,hu2018nearly}. Contrary to our result for the sliding window case, in the fully-dynamic case, obtaining a $\tilde \O(k)$ memory, low update time algorithm, for the {\it arbitrary} metric $k$-clustering case with general $\ell_p$ norms is an open problem. For the special case of $d$-dimensional Euclidean spaces for $k$-means, there are positive results---\cite{hu2018nearly} give $\tilde O(k d^4)$-space core-set with $1+\epsilon$ approximation. 

Dynamic algorithms have also been studied in a consistent model ~\cite{cohen2019fully,lattanzi2017consistent}, but there the objective is to minimize the number of changes to the solution as the input evolves, rather than minimizing the approximation ratio and space used.  Finally, a relaxation of the fully dynamic model that allows only a limited number of deletions has also been addressed~\cite{ginart2019making,mirzasoleiman2017deletion}. The only work related to clustering is that of submodular maximization ~\cite{mirzasoleiman2017deletion} which includes exemplar-based clustering as a special case.

\paragraph{Our Contributions}
\label{sec:contrib}
We simplify and improve the state-of-the-art of $k$-clustering sliding window algorithms, resulting in lower memory algorithms.  Specifically, we:
\begin{itemize}\itemsep=0in\vspace{-0.0in}
\item Introduce a simple new algorithm for $k$-clustering in the sliding window setting (Section~\ref{sec:algo}). The algorithm is an example of a more general technique that we develop for minimization problems in this setting. (Section \ref{sec:algorithm}).
\item Prove that the algorithm needs space linear in $k$ to obtain a constant approximate solution (Theorem~\ref{th:algo-main-th}), thus improving over the best previously known result which required $\Omega(k^3)$ space.
\item Show empirically that the algorithm is orders of magnitude faster, more space efficient, and more accurate than previous solutions, even for small values of $k$ (Section~\ref{sec:experiments}).  
\end{itemize}\vspace{-0.15in}

%% file: prelim.tex
\section{Preliminaries}
\label{sec:preliminaries}

Let $X$ be a set of arbitrary points, and $\dist: X \times X \rightarrow \mathbb{R}$ be a distance function. We assume that $(X,\dist)$ is an arbitrary metric space, that is, $\dist$ is non-negative, symmetric, and satisfies the triangle inequality. 
For simplicity of exposition we will make a series of additional assumptions, \ifappendix
in Appendix~\ref{app:assumption},
\else
 in supplementary material, 
\fi we explain how we can remove all these assumptions.
We assume that the distances are normalized to lie between $1$ and $\Delta$. 
We will also consider weighted instances of our problem where, in addition, we are given a function $\weight: X\rightarrow \mathbb{Z}$ denoting the multiplicity of the point.

The $k$-clustering family of problems asks to find a set of $k$ cluster centers that minimizes a particular objective function. For a point $x$ and a set of points $Y = \{y_1, y_2, \ldots, y_m\}$, we let  $d(x, Y) = \min_{y \in Y} d(x, y)$, and let $\cl(x, Y)$ be the point  that realizes it, $\arg\min_{y \in Y} d(x, y).$ The cost of a set of centers $\centers$ is: $\cost_p(X, \centers) =  \sum_{x \in X} d^p(x, \centers)$. 
Similarly for weighted instances, we have $\cost_p(X, \weight, \centers) =  \sum_{x \in X} \weight(x) d^p(x, \centers)$. %

Note that for $p=2$, this is precisely the \textsc{$k$-Medoids} problem.\footnote{In the Euclidean space, if the centers do not need to be part of the input, then setting $p=2$ recovers the \textsc{$k$-Means} problem.} For $p=1$, the above encodes the \textsc{$k$-Median} problem. When $p$ is clear from the context, we will drop the subscript. We also refer to the optimum cost for a particular instance $(X,d)$ as $\opt_p(X)$, and the optimal clustering as $\centers^*_p(X) =  \{\cen_1^*, \cen_2^*, \ldots, \cen_k^*\}$ , shortening to $\centers^*$ when clear from context.  Throughout the paper, we assume that $p$ is a constant with $p\ge 1$. 

While mapping a point to its nearest cluster is optimal, any map $\mu: X \rightarrow X$ will produce a valid clustering.  In a slight abuse of notation we extend the definition of $\cost_p$ to say $\cost_p(X, \mu) = \sum_{x \in X} d(x, \mu(x))^p.$ 

In this work, we are interested in algorithms for sliding window problems, we refer to the window size as $w$ and to the set of elements in the active window as $W$, and we use $n$ for the size of the entire stream, typically $n \gg w$. We denote by $X_t$ the $t$-th element of the stream and by $X_{[a,b]}$ the subset of the stream from time $a$ to $b$ (both included). For simplicity of exposition, we assume that we have access to a lower bound $m$ and upper bound $M$ of the cost of the optimal solution in any sliding window.\footnote{These assumptions are not necessary.
\ifappendix
  In section~\ref{sect:app-exp}, we explain how we estimate them in our experiments and in appendix~\ref{app:assumption} we show how we can remove them.
\else
 In the supplementary material, we explain how we estimate them in our experiments and how from a theoretical perspective we can remove the assumptions.
\fi
} 

We use two tools repeatedly in our analysis. The first is the relaxed triangle inequality. For $p \geq 1$ and any $x,y,z \in X$, we have:
$d(x,y)^p\leq 2^{p-1} (d(x,z)^p + d(z,y)^p).$ The second is the fact that the value of the optimum solution of a clustering problem does not change drastically if the points are shifted around by a small amount. This is captured by Lemma~\ref{lem:cmap} which was first proved in \cite{Guha:2000:CDS:795666.796588}. For completeness we present its proof in 
\ifappendix
 Appendix~\ref{app:prelim}.
\else 
 the supplementary material.
\fi 
\begin{lemma}\label{lem:cmap}
Given a set of points $X = \{x_1, \ldots, x_n\}$ consider a multiset $Y = \{y_1, \ldots, y_n\}$ such that $\sum_i d^p(x_i, y_i) \leq \alpha \opt_p(X)$, for a constant $\alpha$. Let $\mathcal{B}^*$ be the optimal $k$-clustering solution for $Y$. Then $\cost_p(X, \mathcal{B^*}) \in O((1+\alpha)OPT_p(X))$.
\end{lemma}

Given a set of points $X$, a mapping $\mu: X \rightarrow Y$, and a weighted instance defined by $(Y, \weight)$, we say that the weighted instance is {\em consistent} with $\mu$, if for all $y\in Y$, we have that $\weight(y) = |\{x\in X | \; \mu(x) = y\}|$. We say it is {\em $\epsilon$-consistent} (for constant $\epsilon \geq 0$), if for all $y\in Y$, we have that $|\{x\in X \;|\; \mu(x) = y\}| \leq \weight(y) \leq (1+\epsilon) |\{x\in X \;|\; \mu(x) = y\}|$.

Finally, we remark that the $k$-clustering problem is NP-hard, so our focus will be on finding efficient approximation algorithms. We say that we obtain an $\alpha$ approximation for a clustering problem if  $\cost_p(X, \centers) \leq \alpha \cdot \opt_p(X)$. The best-known approximation factor for all the problems that we consider are constant ~\cite{DBLP:conf/focs/AhmadianNSW17,DBLP:conf/soda/ByrkaPRST15,DBLP:journals/corr/abs-0809-2554}.
Additionally, since the algorithms work in arbitrary metric spaces, we measure update time in terms of distance function evaluations and use the number of points as space cost (all other costs are negligible).

%% file: algo.tex
\subsection{Sliding Window Algorithm}
\label{sec:algo}
In the previous section we have shown that we can the Meyerson sketch to have enough information to output a solution using the points in the active window whose cost is comparable to the cost of the optimal computed on the whole stream. However, we need an algorithm that is competitive with the cost of the optimum solution computed solely on the elements in the sliding window. 

We give some intuition behind our algorithm before delving into the details. Suppose we had a good guess on the value of the optimum solution, $\lambda^*$ and imagine splitting the input $x_1, x_2, \ldots, x_t$ into blocks $A_1 = \{x_1,x_2, \ldots, x_{b_1}\}$, $A_2 = \{x_{b_1+1}, \ldots, x_{b_2}\}$, etc. with the constraints that (i) each block has optimum cost smaller than $\lambda^*$, and (ii) is also maximal, that is adding  the next element to the block causes its cost to exceed $\lambda^*$. It is easy to see, that any sliding window of optimal solution of cost $\lambda^*$ overlaps at most two blocks. The idea behind our algorithm is that, if we started an augmented Meyerson sketch in each block, and we obtain a good mapping for the suffix of the first of these two blocks, we can recover a good approximate solution for the sliding window.

We now show how to formalize this idea. 
During the execution of the algorithm, we first discretize the possible values of the optimum solution, and run a set of sketches for each value of $\lambda$.
Specifically, for each guess $\lambda$, we run Algorithm~\ref{alg:pseudo} to compute the $\textrm{AugmentedMeyerson}$ for two consecutive substreams, $A_\lambda$ and $B_\lambda$, of the input stream $X$. (The full pseudocode of $\textrm{AugmentedMeyerson}$ is available in \refappendix{app:AugmentedMeyerson}.)
When a new point, $x$, arrives we check whether the $k$-clustering cost of the solution computed on the sketch after adding $x$ to $B_\lambda$ exceeds $\lambda$. If not, we add it to the sketch for $B_\lambda$, if so we reset the $B_\lambda$ substream to $x$, and rename the old sketch of $B_\lambda$ as $A_\lambda$.  Thus the algorithm maintains two sketches, on consecutive subintervals. Notice that the cost of each sketch is at most $\lambda$, and each sketch is grown to be maximal before being reset. 

We remark that to convert the Meyerson sketch to a $k$-clustering solution, we need to run a $k$-clustering algorithm
 on the weighted instance given by the sketch. Since the problem is NP-hard, let $\algo$ denote any $\rho$-approximate algorithm, such as the one by~\cite{DBLP:journals/corr/abs-0809-2554}. Let $S(Z) = (Y(Z), \weight(Z))$ denote the augmented Meyerson sketch built on a (sub)stream $Z$, with $Y(Z)$ as the centers, and $\weight(Z)$ as the (approximate) weight function. We denote by $\algo(S(Z))$  the solution obtained by running $\algo$ over the weighted instance $S(Z)$.  Let $\acost_p(S(Z), \algo(S(Z)))$ be the estimated cost of the solution $\algo(S(Z))$ over the stream $Z$ obtained by the sketch $S(Z)$.  
 
 We show that we can implement a function $\acost_p$ that operates only on the information in the augmented Meyerson sketch $S(Z)$ and gives a $\beta \in O(\rho)$ approximation to the cost on the unabridged input.
 
\begin{lemma}[Approximate solution and approximate cost from a sketch]\label{lem:aas}
Using an approximation algorithm $\algo$, from the augmented Meyerson sketch $S(Z)$, with probability $\ge 1-\gamma$, we can output a solution $\algo(S(Z))$ and an estimate  $\acost_p(S(Z), \algo(S(Z)))$ of its cost s.t. 
$
\cost_p(Z, \algo(S(Z)))   \le  \acost_p(S(Z), \algo(S(Z)))
  \le \beta(\rho) \cost_p(Z, \opt(Z))
$
for a constant $\beta(\rho) \le 2^{3p+6}\rho$ depending only the approximation factor $\rho$ of $\algo$. 
\end{lemma}
 
\begin{footnotesize}
\begin{algorithm}
\small
\caption{\small Meyerson Sketches, $ComputeSketches(X, w,\lambda, m, M, \Delta)$}\label{alg:pseudo}
\begin{algorithmic}[1]
\STATE Input: A sequence of points $X=x_0,x_1,x_2,\dots,x_n$. The size of the window $w$. Cost threshold $\lambda$. A lower bound $m$ and upper bound $M$ of the cost of the optimal solution and upper bound on distances $\Delta$.
\STATE Output: Two sketches for the stream $S_1$ and $S_2$.
\STATE  $S_1 \gets \textrm{AugmentedMeyerson}(\emptyset, w, m, M, \Delta)$; $S_2 \gets \textrm{AugmentedMeyerson}(\emptyset, w, m, M, \Delta)$\; 
\STATE  $A_{\lambda} \gets \emptyset$; $B_{\lambda} \gets \emptyset$\; (Recall that $A_{\lambda}, B_{\lambda}$ are sets and $S_1$ and 
$S_2$ the corresponding sketches. Note that the content of the sets is not stored explicitly.)
  \FOR{$x \in X$}{
   \STATE Let $S_{temp}$ be computed by $\textrm{AugmentedMeyerson}(B_{\lambda}\cup \{x\}, w, m, M, \Delta)$ . (Note: it can be computed by adding $x$ to a copy of the sketch maintained by $S_2$) 
    \IF{$\acost_p(S_{temp}, \algo(S_{temp})) \leq \lambda$}{
      \STATE Add $x$ to the stream of the sketch $S_2$.\; ($B_{\lambda} \gets B_{\lambda}\cup \{x\},\; S_2 \gets \textrm{AugmentedMeyerson}(B_{\lambda}, w, m, M, \Delta)$) 
    } \ELSE {
      \STATE $S_1 \gets S_2$; $S_2 \gets \textrm{AugmentedMeyerson}(\{x\}, w, m, M, \Delta)$.\; ($A_{\lambda} \gets B_{\lambda}$; $B_{\lambda} \gets \{x\}$)
    }
    \ENDIF
  }
  \ENDFOR
  \STATE Return ($S_1, S_2$, and start and end times of $A_\lambda$ and $B_\lambda$)
\end{algorithmic}
\end{algorithm}
\end{footnotesize}

\paragraph{Composition of sketches from sub-streams}
Before presenting the global sliding window algorithm that uses these pairs of sketches, we introduce some additional notation. Let $S(Z)$ be the augmented Meyerson sketch computed over the stream $Z$. Let $\suffix_\tau(S(Z))$ denote the sketch obtained from a sketch $S$ for the points that arrived after $\tau$. This can be done using the operations defined in \refappendix{app:suffix-sketch}. 

We say that a time $\tau$ is contained in a substream $A$ if $A$ contains elements inserted on or after time $\tau$. Finally we define $A_\tau$ as the suffix of $A$ that contains elements starting at time $\tau$.
Given two sketches $S(A)$, and $S(B)$ computed over two disjoint substreams $A,B$, let $S(A) \cup S(B)$ be the sketch obtained by joining the centers of $S(A)$ and $S(B)$ (and summing their respective weights) in a single instance. We now prove a key property of the augmented Meyerson sketches we defined before. 
\begin{lemma}[Composition with a Suffix of stream] \label{lem:suffix-composition}
 Given two substreams $A$,$B$ (with possibly $B=\emptyset$) and a time $\tau$ in $A$, let $\algo$ be a constant approximation algorithm for the $k$-clustering problem. Then if $\opt_p(A) \le O(\opt_p(A_\tau \cup B)$, then, with probability $\ge 1-O(\gamma)$, we have  $f_p(A_\tau \cup B, \algo(\suffix_\tau(S(A)) \cup S(B))) \le O(\opt_p(A_\tau \cup B))$.
\end{lemma} 
The main idea of the proof is to show that 
 $\suffix_\tau(S(A)) \cup S(B)$ is $\epsilon$-consistent with a good mapping from $A_\tau \cup B$ and then by using a technique similar to  Lemma~\ref{lem:cmap}  show that we can compute a constant approximation from an $\epsilon$-consistent sketch.

\paragraph{Final algorithm.}

\begin{footnotesize}
\begin{algorithm}
\small
\caption{Our main algorithm. Input: $X,m,M,\Delta$,  approx. factor of $\algo$ ($\beta$) and $\delta$.}
\label{alg:algo-min}
\begin{algorithmic}[1]
\STATE $\Lambda \gets \{m,(1+\delta) m,\dots,2^p\beta(1+\delta)M\}$\;
\FOR  {$\lambda \in \Lambda$}
 \STATE $S_{\lambda, 1}, S_{\lambda, 2} \gets \mathrm{ComputeSketches}(X, w, \lambda, m, M, \Delta)$
\ENDFOR
\STATE {\bf if } {$B_{\lambda^*}=W$ for some $\lambda^*$}
  {\bf then return}  $\algo(S_{\lambda^*, 2})$
%\ENDIF
\STATE $\lambda^*\gets \min(\{\lambda:A_\lambda \not\subseteq W\})$
\STATE $\tau \gets \max(|X|-w, 1)$

\STATE {\bf if } {$W \cap A_{\lambda^*} \neq \emptyset$}
 {\bf then return} $\algo(\suffix_\tau(S_{\lambda^*, 1} ) \cup S_{\lambda^*, 2})$  \\
%\ELSE
\STATE {\bf else}  {\bf  return} $\algo(\suffix_\tau(S_{\lambda^*, 2}))$
%\ENDIF
\end{algorithmic}
\end{algorithm}
\end{footnotesize}

We can now  present the full algorithm in Algorithm ~\ref{alg:algo-min}.  As mentioned before, we run multiple copies of $ComputeSketches$ in parallel, for geometrically increasing values of  $\lambda$. 

For each value of  $\lambda$, we maintain the pair of sketches over the stream $X$.  Finally, we compute the centers using such sketches. If we get lucky, and for the sliding window $W$ there exists a subsequence where $B_{\lambda^*}$ is precisely $W$, we use the appropriate sketch and return $\algo(S_{\lambda^*,2})$. Otherwise, we find the smallest $\lambda^*$ for which $A_\lambda$ is not a subset of $W$. We then use the pair of sketches associated with $A_{\lambda^*}$ and $B_{\lambda^*}$, combining the sketch of the suffix of $A_{\lambda*}$ that intersects with $W$, and the sketch on $B_{\lambda^*}$. 

The main result is that this algorithm provides a constant approximation of the $k$-clustering problem, for any $p\ge 1$, with probability at least $1-\gamma$, using space linear in $k$ and logarithmic in other parameters. The total running time of the algorithm depends on the complexity of $\algo$. Let $T(n,k)$ be the complexity of solving an instance of $k$-clustering with size $n$ points using $\algo$. 

\begin{theorem}\label{th:algo-main-th}
With probability $1-\gamma$, Algorithm~\ref{alg:algo-min}, outputs an $O(1)$-approximation
for the sliding window $k$-clustering problem using space:
$O\big(k \log (\Delta) (\log (\Delta) + \log(w) + \log(M) )$
$\log^2 (M / m) \log(\gamma^{-1} \log(M/m)) \big)$
and total update time 
$O(  T(k \log (\Delta), k)$ $\log^2 (M / m) \log(\gamma^{-1}  \log(M / m))$ $ (\log (\Delta) + \log(w) + \log(M) )$.
\end{theorem}

We remark that if $M$ and $\Delta$ are polynomial in $w$, then the total space is $O(k \log^4 w \log (\log w/\gamma))$ and the total update time is $O(T(k \log w, k) \log^3(w)  \log (\log w/\gamma))$. The main component in the constant approximation factor of Theorem~\ref{th:algo-main-th} statement comes from the $2^{3p+5}\rho$ approximation for the insertion-only case~\cite{lattanzi2017consistent}. Here $p$ is the norm, and $\rho$ is the offline algorithm factor. Given the composition operation in our analysis in addition to applying triangle inequality and some other steps, we end up with an approximation factor $\approx 2^{8p+6}\rho$. We do not aim to optimize for this approximation factor, however it could be an interesting future direction.

%% file: experiments.tex
\section{Empirical Evaluation}
\label{sec:experiments}

We now describe the methodology of our  empirical evaluation before providing our experiments results.
We report only the main results in the section, more details on the experiments and results are \ifappendix 
 in Appendix~\ref{sect:app-exp}. \else 
 in supplementary material.
 \fi
Our code is available {open-source} on github\footnote{\url{https://github.com/google-research/google-research/tree/master/sliding_window_clustering/}}. All datasets used are {\em publicly-available}.

\textbf{Datasets.}
We used 3 real-world datasets from the UCI Repository~\cite{Dua:2017} that have been used in previous experiments on $k$-clustering for data streams settings: SKINTYPE~\cite{bhatt2010skin}, $n=245057, d=4$,  SHUTTLE,  $n=58000, d=9$, and COVERTYPE~\cite{blackard1999comparative},  $n=581012, d=54$. Consistent with previous work, we stream all points in the natural order (as they are stored in the dataset). 
We also use 4 publicly-available synthetic dataset from~\cite{ClusteringDatasets} (the S-Set series) that have ground-truth clusters. We use 4 datasets (s1, s2, s3, s4) that are increasingly harder to cluster and have each $k=15$ ground-truth clusters. Consistent with previous work, we stream the points in random order (as they are sorted by ground truth in the dataset). 
In all datasets, we pre-process each dataset to have zero mean and unit standard deviation in each dimension.  All experiments use  Euclidean distance, we focus on the  the \textsc{k-Means} objective ($p=2$) which we use as cost. We use $k$-means++~\cite{arthur2007k} as the solver $\algo$ to extract the solution from our sketch.

\textbf{Parameters.} We vary the number of centers, $k$, from $4$ to $40$ and window size, $w$,  from $10{,}000$ to $40{,}000$. We experiment with $\delta  = [0.1, 0.2]$ and set $\epsilon = 0.05$ (empirically the results are robust to wide settings of $\epsilon$).

\textbf{Metrics.}
We focus on three key metrics: cost of the clustering, maximum space requirement of our sketch, and average running time of the update function. To give an implementation independent view into space and update time, we report as space usage the number of points stored, and as update time the number of distance evaluations. All of the other costs are negligible by comparison.

\begin{figure}
\begin{center}
\includegraphics[width=0.40\textwidth,keepaspectratio]{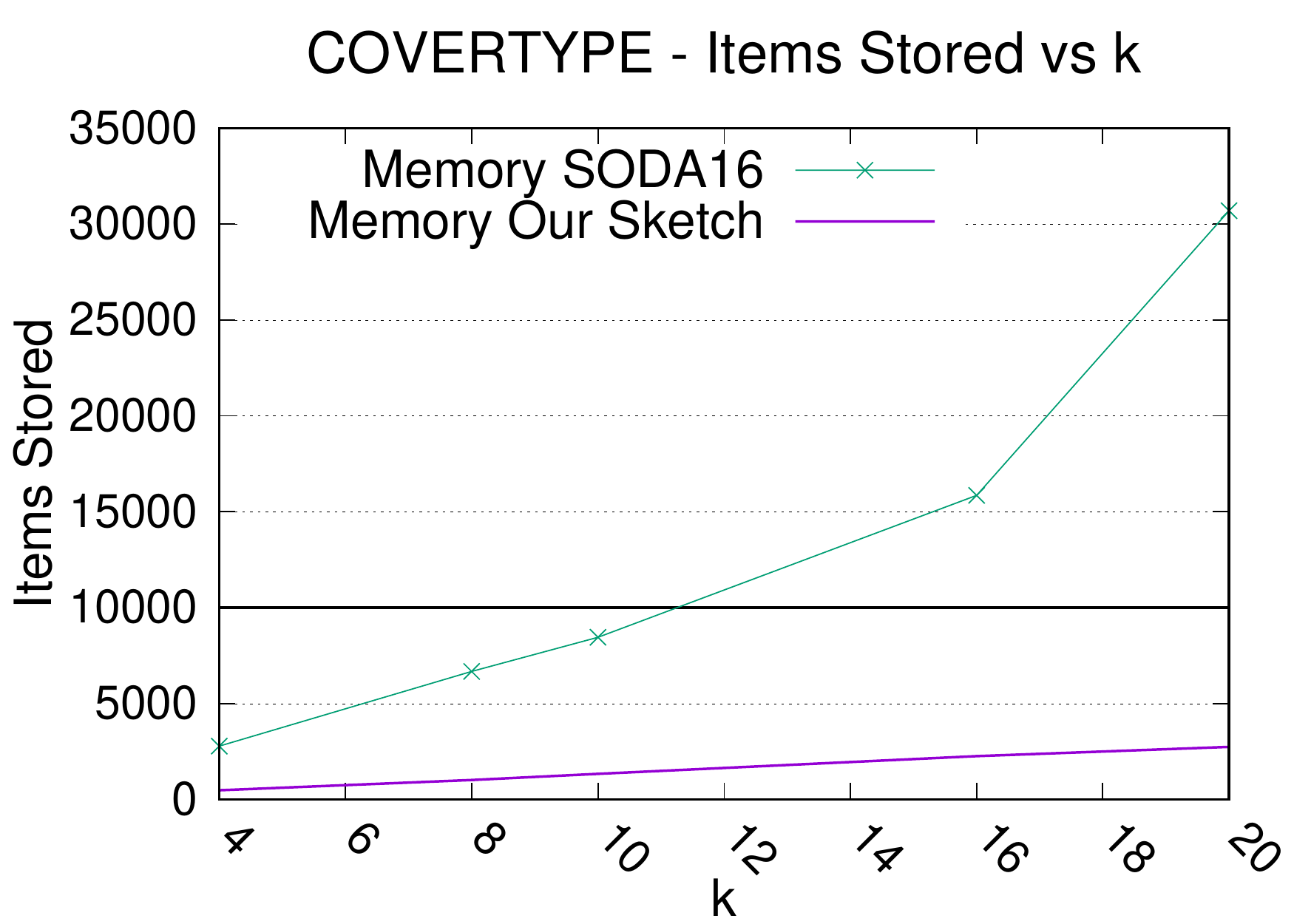}
\caption{Comparison of the max memory use of SODA16 and our algorithm for $W=10{,}000$. }
\label{fig:soda}
\end{center}
\end{figure}

\textbf{Baselines.} We consider the following baselines.\\
\textbf{Batch K-Means++}:  We use $k$-means++ over the entire window as a proxy for the optimum, since the latter is NP-hard to compute. At every insertion, we report the best solution over 10 runs of $k$-means++ on the window.  Observe that this is inefficient as it requires $\Omega(w)$ space and $\Omega(kw)$ run time per update. 
\textbf{Sampling}: We maintain a random sample of points from the active window, and then run $k$-means++ on the sample. This allows us to evaluate the performance of a baseline, at the same space cost of our algorithm. 
\textbf{SODA16}: We also evaluated the only previously published algorithm for this setting in~\cite{DBLP:conf/soda/BravermanLLM16}. 

We note that we made some practical modifications to further improve the performance of our algorithm which we report \ifappendix 
in the Appendix. 
\else 
in the supplementary material.
\fi 

\begin{table*}
\small
\centering
\begin{tabular}{ c | c | c | c | c |}
\textbf{Dataset} & $k$ &
\makecell{\textbf{Space Decr.} \\ \textbf{Factor}} &
\makecell{\textbf{Speed-Up} \\ \textbf{Factor}}  &
\makecell{\textbf{Cost} \\ \textbf{(ratio)}}  \\
\hline
\multirow{2}{*}{COVER}&4&5.23&10.88&$99.5\%$\\
 &16&6.04&13.10&$95.1\%$\\
\hline
\multirow{2}{*}{SHUTTLE}&4&5.07&9.09&$106.8\%$\\
&16&15.32&15.14&$118.9\%$\\
\end{tabular}
\caption{Decrease of space use, decrease in (speed-up) and ratio of mean cost of the solutions of our algorithm vs the SODA16 baseline  ($100\%$ means same cost, $< 100\%$ means a reduction in cost). 
}
\label{table:soda16}
\end{table*}

\paragraph{Comparison with previous work.}
We begin by comparing our algorithm to the previously published algorithm of~\cite{DBLP:conf/soda/BravermanLLM16}.  
The baseline in this paragraph is \textbf{SODA16} algorithm in \cite{DBLP:conf/soda/BravermanLLM16}.
We confirm empirically that the memory use of this baseline already exceeds the size of the sliding window for very small $k$, and that it is significantly slower than our algorithm. Figure~\ref{fig:soda} shows the space used by our algorithm and by the baseline over the COVERTYPE dataset for a $|W|=10{,}000$ and different $k$. We confirm that our algorithm's memory grows linearly in $k$ while the baseline grows super-linearly in $k$ and that for $k>10$ the baseline costs more than storing the entire window.
In Table~\ref{table:soda16} we show that our algorithm is significantly faster and uses less memory than the \textbf{SODA16} already for small values of $k$. 
\ifappendix
In the appendix  
\else
In the supplementary material 
\fi 
we show that the difference is even larger for bigger values of $k$.
Given the inefficiency of the SODA16 baseline, for the rest of the section we do not run experiments with it.

\begin{table*}%{r}{.70\linewidth}
\begin{center}
\begin{tabular}{ c | c | c | c |}
$W$ & $k$ & \textbf{Space}  & \textbf{Time}  \\
\hline
\multirow{3}{*}{40000}&10&$3.5\%$ $ (0.39\%)$&$0.45\%$ $ (0.29\%)$\\
&20&$6.5\% $ $ (0.87\%)$&$0.93\% $ $(0.63\%)$\\
&40&$11.3\%$ $ (1.74\%)$&$1.58\%$ $ (1.23\%)$\\
\end{tabular}
\caption{\small Max percentage of sliding window (length $W$) stored (Space) and median percentage of time (Time) vs. one run of $k$-means++ (stddev in parantesis).}
\label{table:time-space-over-baseline}
\end{center}
\end{table*}

\begin{figure}
\begin{centering}
\includegraphics[width=0.45\textwidth,keepaspectratio]{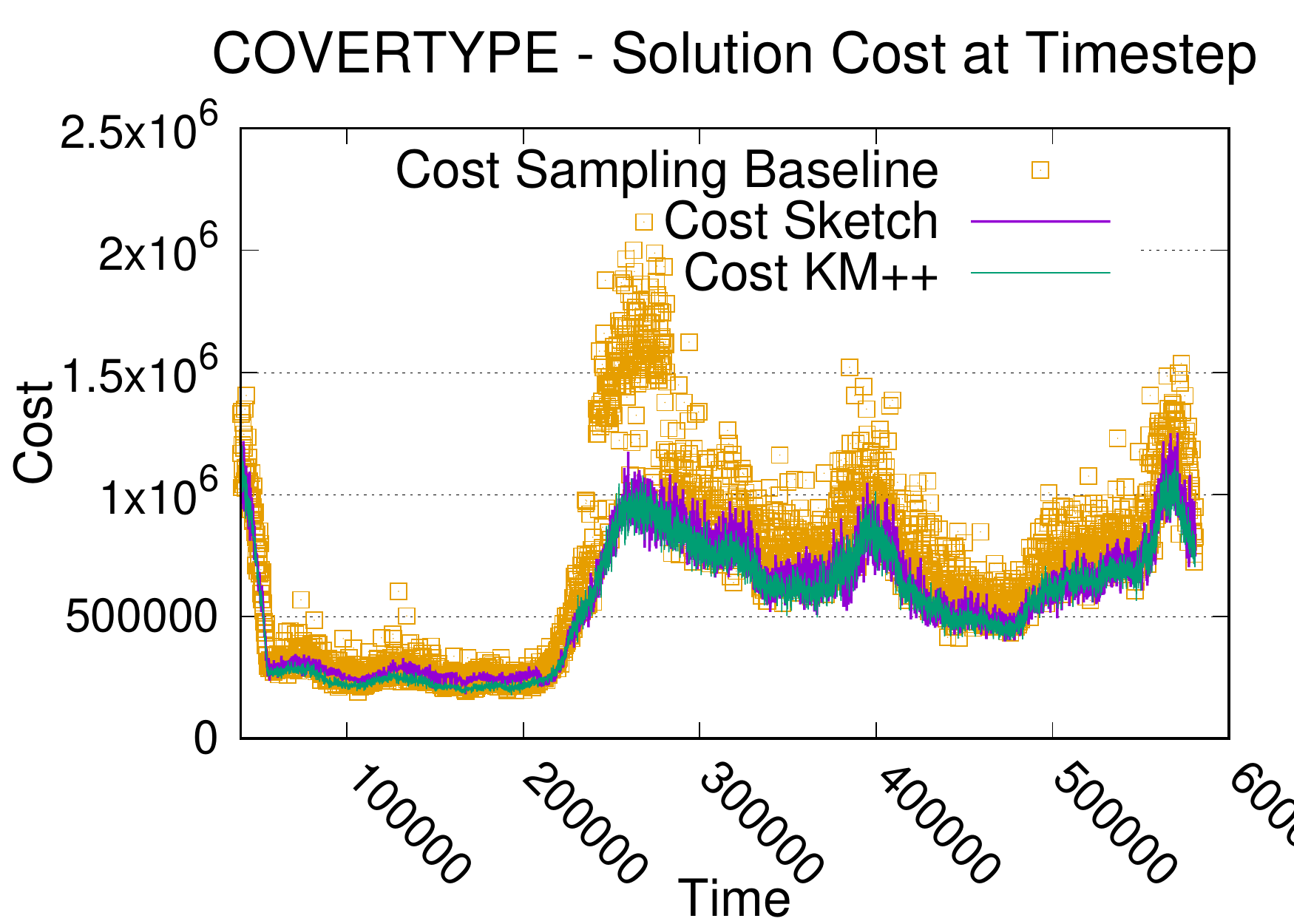}
\caption{Cost of the solution obtained by our algorithm (Sketch) and the two baselines for $k=20$, $|W|=40{,}000$ and $\delta=0.2$ on COVERTYPE.}
\label{fig:evolution-covtype}
\end{centering}
\end{figure} 

\paragraph{Cost of the solution.} 
We now take a look at how the cost of the solution evolves over time during the execution of our algorithm.
In Figure \ref{fig:evolution-covtype} we plot the cost of the solution obtained by our algorithm (Sketch), our proxy for the optimum (KM++) and the sampling baseline (Sampling Baseline) on the COVERTYPE dataset. The sampling baseline is allowed to store the same number of points stored by our algorithm (at the same point in time). We use $k=20$, $|W|=40{,}000$, and $\delta=0.2$. The plot is obtained by computing the cost of the algorithms every $100$ timesteps. Observe that our algorithm closely tracks that of the offline algorithm result, even as the cost fluctuates up and down. Our algorithm's cost is always close to that of the off-line algorithm and significantly better than the random sampling baseline

\paragraph{Update time and space tradeoff.}

We now investigate the time and space tradeoff of our algorithm. As a baseline we look at the cost required simply to recompute the solution using $k$-means++ at every time step. 
In Table~\ref{table:time-space-over-baseline} ($\delta=0.2$) we focus on the COVERTYPE dataset, the other results are similar. Table~\ref{table:time-space-over-baseline} shows the percent of the sliding window data points stored (Space) and the percent of update time (Time) of our algorithm vs a single run of $k$-means++ over the window. 
\ifappendix 
The Table~\ref{table:time-space-over-baseline2} in Appendix~\ref{sect:app-exp} shows that the savings become larger (at parity of $k$) as $|W|$ grows and that we always store a small fraction of the window, providing order-of-magnitude speed ups (e.g., we use $<0.5\%$ of the time of the baseline for $k=10$, $|W|=40{,}000$). Here the baseline is the $k$-means++ algorithm.
\else 
In the supplementary material we show that the savings become larger (at parity of $k$) as $|W|$ grows and that we always store a small fraction of the window, providing order-of-magnitude speed ups (e.g., we use $<0.5\%$ of the time of the baseline for $k=10$, $|W|=40{,}000$). Here the baseline is the $k$-means++ algorithm.
\fi

\paragraph{Recovering ground-truth clusters.}
We evaluated the accuracy of the clusters produced by our algorithm on a dataset with ground-truth clusters using the well known V-Measure accuracy definition for clustering~\cite{rosenberg2007v}. We observe that on all datasets our algorithm performs better than the sampling baseline and in line with the offline $k$-means++. For example, on the s1 our algorithm gets V-Measure of 0.969, while $k$-means++ gets $0.969$ and sampling gets  $0.933$. The full results are available \ifappendix in the appendix~\ref{sect:app-exp}. \else in the supplementary material. \fi

%% file: app-known.tex
\section{Omitted Proofs from Preliminaries}
\label{app:prelim}

\begin{lemma}[Restatement of Lemma~\ref{lem:cmap}]
 Given a set of points $X = \{x_1, \ldots, x_n\}$ consider a multiset $Y = \{y_1, \ldots, y_n\}$ such that $\sum_i d^p(x_i, y_i) \leq \alpha \opt_p(X)$, for a constant $\alpha$. Let $\mathcal{B}^*$ be the optimal solution for $Y$. Then $\cost_p(X, \mathcal{B^*}) \in O(OPT_p(X))$.
\end{lemma}

\begin{proof}
This proof follows through several applications of triangle inequality.  First, observe that:
\begin{equation}
\label{eqn:tri1}
\cost_p(X, \mathcal{B}^*) = \sum_i d^p(x_i, \mathcal{B^*}) \leq 
2^{p-1} \sum_i \Big(d^p(x_i, y_i) + d^p(y_i, \mathcal{B^*})\Big)
\end{equation}

instance $(X,d)$ as $\opt_p(X)$, and the optimal clustering as $\centers^*_p(X) =  \{\cen_1^*, \cen_2^*, \ldots, \cen_k^*\}$ , shortening to $\centers^*$ when clear from context.  Notice that everywhere in the paper, we assume that $p$ is a constant with $p\ge 1$.

Let $\centers^* = \{\cen_1^*, \ldots, \cen_k^*\}$ be the optimal solution for $X$. In a slight abuse of notation, let $\centers^*(x_i)$ represent the optimal center closest to $x_i$. In addition, for every $\cen_j^*$, let $t_j$ be its closest neighbor in $Y$, $t_j = \cl(\cen_j^*, Y)$. Observe that the set $T = \{t_1, \ldots, t_k\}$ represents a feasible solution. 

Therefore:
\begin{equation}
\label{eqn:tri2}
\sum_i d^p(y_i, B^*) \leq d^p(y_i, T).
\end{equation}

Finally, let $t(i) = \cl(\centers^*(x_i), Y)$ represent the point in $Y$ which is closest to the optimal center assigned to $x_i$.

\begin{align*}
\sum_i d^p(y_i, \mathcal{B^*}) &\leq \sum_i d^p(y_i, t(i)) \\
&\leq 2^{p-1} \sum_i \Big(d^p(y_i, \centers^*(x_i)) + d^p(\centers^*(x_i), t(i))\Big) \\
&\leq 2^{p-1} \sum_i 2 \cdot d^p(y_i, \centers^*(x_i)) \\
&\leq 2^{2p-1} \Big(\sum_i d^p(y_i, x_i) +  d^p(x_i, \centers^*(x_i))\Big)\\
&\leq 2^{2p-1} (1+\alpha) OPT_p(X),
\end{align*}
Where the first inequality follows because $T$ is a feasible solution, the second and fourth by the generalized triangle inequality, and the third by the definition of $t(i)$. Putting this together with Equation \ref{eqn:tri1} completes the proof. 
\end{proof}

For completeness we present the next lemma that would be useful to show our main theorem

\begin{lemma}\label{lem:adding}
Given a set of points $X = \{x_1, \ldots, x_n\}$ by adding points to the set $X$ the optimal cost of the $k$-clustering can decrease at most by a  $2^p$.
\end{lemma}
\begin{proof}
Let $Y = \{x_1, \ldots, x_n, y_1,\ldots, y_m\}$ be the set obtained by adding $\{y_1,\ldots, y_m\}$ arbitrary points to X. Let $ \centers^*=\{\cen_1^*, \cen_2^*, \ldots, \cen_k^*\}$ be the optimal set of centers for $Y$ and let $\centers' = \{x_{\cen_1^*},\ldots,x_{\cen_k^*}\}$ the (multi-)set where $x_{\cen_i^*}$ is the closest point to  $\cen_i^*$ in $X$. Furthermore let $\cen^*(x)$ be the closest center in  $ \centers^*$ to a point $x$. Then we have $\opt_p(X) \leq \sum_{x\in X}d(x, \centers')\leq\sum_{x\in X}2^{p-1}(d(x, \cen^*(x))+d(\cen^*(x), x_{\cen^*(x)})\leq \sum_{x\in X}2^{p}d(x, \cen^*(x)) \leq 2^p \sum_{x\in X}d(x,  \centers^*)\leq 2^p \sum_{y\in Y}d(y,  \centers^*) \leq 2^p \opt_p(Y)$.
\end{proof}

\section{Meyerson Sketches}
 
In this section we define the Meyerson sketch algorithm and show some basic facts about the property of any mapping $\mu$ and properties of the Meyerson sketch.

The results presented here follow from well-known results~\cite{charikar03better, lattanzi2017consistent}, but we present them here for completeness. 

\subsection{Meyerson algorithm}

We begin by formally stating the guarantees of the Meyerson sketch, giving a constant factor approximation for any constant $p$.\footnote{We do not optimize constants throughout the paper to keep proofs simpler. We empirically evaluate the algorithm's performance and show the efficacy of our approach.}
The following Lemma, captures the main properties of the sketch.

\begin{lemma}\label{lemma:corollary-of-lemma5o1}
Let $(X,\dist)$ be a metric space and fix $\gamma \in (0,1)$. Then the Meyerson sketch algorithm computes a mapping $\mu: X \rightarrow \centers$, and
a consistent weighted instance $(\centers$,$\weight$), such that $\centers \subseteq X$ and, with probability $1- \gamma$:
$
|\centers| \leq 2^{2p+8} k \log  \nicefrac{1}{\gamma}\log \Delta {\text{;}} \quad 
\cost_p(X, \centers)  \le  2^{p+7} \opt_p(X).$
The algorithm uses at most  $2^{2p+8} k \log  \nicefrac{1}{\gamma}$ $\log \Delta \lceil\log M/m\rceil$ space and in the consistent mapping $\mu$, each point is mapped to a center inserted before it into the stream. 
\end{lemma}

We now present the pseudocode of the Meyerson sketch algorithm. Fix the parameter $p$ of the problem. Consider for now to have access to a lower bound to the cost $L^p$ of the optimal solution for the problem, such that $L^p \geq \alpha \opt_p$, for some constant $0 < \alpha < 1$. We will remove this assumption at the end of the section.

We start by presenting the code for a single Meyerson sketch as described in Algorithm~\ref{alg:meyerson}. 
It picks each point as a new center with the probability proportional to the distance of the point from the selected centers so far. 
We note that a point that is far from the current centers should be be added to avoid a large cost for that point. 
To obtain the results in Lemma~\ref{lemma:corollary-of-lemma5o1} we will run multiple instances of Algorithm~\ref{alg:meyerson}.

\begin{footnotesize}
\begin{algorithm}[h!]
\small
\caption{\small Single Meyerson sketch}\label{alg:meyerson}
\begin{algorithmic}[1]
\STATE Input: A sequence of points $x_0,x_1,x_2,\dots,x_n$. A finite $p$, a guess $L^p$, an upper bound $\Delta$ of the max distance, and the parameter $k$.
\STATE Output: A mapping $\mu: X \rightarrow S$, and a weighted instance $(S, \weight)$ that is consistent with $X$ and cost of the moving $\costm$.
\STATE $S \leftarrow \emptyset$
\STATE Let $X$ be a set of points and assume $L^p$ is such that $L^p \geq \alpha \opt_p(X)$, for some constant $0 < \alpha < 1$
\STATE $\costm = 0$
\FOR{$x\in X$}
	 \IF{$S = \emptyset$}
	 	\STATE $S \leftarrow \{x\}$
	\ELSE
      	      \STATE With probability $\min\left(\frac{k (1 + \log \Delta)d(x,S)^p}{L^p},1\right)$ add $x$ to $S$ and set $\weight(x)=1$
	      \STATE Otherwise, let $z \leftarrow \argmin_{y\in S}d(x,y)$, set $\weight(z) \leftarrow  \weight(z)+1$, $\costm \leftarrow  \costm + d(x,y)^p$
	\ENDIF
\ENDFOR
\STATE Return $S$, $\weight$, $\costm$
\end{algorithmic}
\end{algorithm}
\end{footnotesize}

Using a single sketch we can show the following property. 

\begin{lemma}\label{lem:meyerson}
For a constant $\alpha \in (0, 1)$ and input $X$, Algorithm~\ref{alg:meyerson} computes a mapping $\mu$ and a consistent weighted instance $(S,\weight)$ with probability at least $\frac12$:
\begin{align*}
& |S| \leq 4 k (1+\log \Delta) \left(\frac{2^{p+3}}{\alpha^p}+1\right) \qquad {\text{and}} \\ 
& \cost_p(X, S)\leq \costm \leq 2^{p+5}\opt_p(X).
\end{align*}
 Furthermore, in the consistent mapping $\mu$ between $X$ and $S$, every point in $X$ is mapped to a point in $S$ inserted before it into the stream. The algorithm computes the cost of the mapping $\costm$.
\end{lemma}

\begin{proof}
Our mapping $\mu$ is defined by $z=\argmin_{y\in S}d(x,y)$ in Algorithm~\ref{alg:meyerson}, note that every point is mapped to a point inserted earlier than it.
We will show that in expectation $S$ has size $k (1+\log \Delta) \left(\frac{2^{2p+1}}{\alpha^p}+1\right)$ and $\cost_p(X, S)\leq 2^{p+3}\opt_p(X)$ . The lemma then follows by an application of Markov inequality. 
Let $\cen_1^*,\cen_2^*,\dots,\cen_k^*$ be the centers in the optimal solution, $C_1^*,C_2^*,\dots,C_k^*$ be the respective clusters and $\calC^*$ be the optimal clustering. Let $\cost_p(C_i^*, \centers^*)$ be the  cost for cluster $i$ and $a_i^*$ be: 
$a_i^* = \left(\frac{\cost_p(C_i^*, \centers^*)}{|C_i^*|}\right)^{\nicefrac1p}.$

Consider the set of points in $C_i^*$ with distance at most $a_i^*$ to the center $\cen_i^*$. Let $D=\{v_1,v_2,\dots\,v_q\}$ be the set of points with distance at most $a_i^*$ to the center $\cen_i^*$. Let $p_i$ be the probability that element $v_i$ is added to $S$, i.e., $p_i=\min(1, \frac{d(v_i,S)^pk(1+\log \Delta)}{L^p})$. With some abuse of notation we say that $p_{q+1}=1$. We want to estimate the cost of the subset of $D$ in the stream before any point in $D$ is added to $S$. For the points $v_1, \ldots, v_{q'}$, $q'\le q$, before the first point added to the set, the probability of them being added to the set is $\frac{d(v_i,S)^pk(1+\log \Delta)}{L^p}<1$ and its cost is $d(v_i,S)^p$. So the expected cost for these points is:
\begin{align*}
& \sum_{i=1}^{q'} (\sum_{j=1}^i d(v_j,S)^p)\left(p_{i+1}\prod_{j=1}^i\left(1 - p_j\right) \right) \\
& = \sum_{i=1}^{q'} d(v_i,S)^p \left(\prod_{j=1}^i\left(1 - p_j\right)\right) \left(\sum_{j=i+1}^{q'+1}p_j\prod_{l=i+1}^{j-1}\left(1 - p_l\right)\right) \\
& \leq   \sum_{i=1}^{q'} d(v_i,S)^p\left(\prod_{j=1}^i\left(1 - p_j\right)\right) \\
& \leq\frac{L^p}{k(1+\log \Delta)} \sum_{i=1}^{q'}\prod_{j=1}^i\left(1 - p_j\right) 
\leq  \frac{L^p}{k(1+\log \Delta)}.
\end{align*} 
The first inequality follows from the fact that we are dividing by a number smaller than $1$ ($\frac{d(v,S)^pk(1+\log \Delta)}{L^p}< 1$).

After we add a point to the set, we can use the generalized triangle inequality to bound the cost of any point $v$ in the set with  $\left(2^{p-1}((a^*_i)^p + d(v,\cen_i^*)^p)\right)$
Therefore,  the probability of adding $v$ as a center is bounded by  $\left(2^{p-1}((a^*_i)^p + d(v,\cen_i^*)^p)\right)$ $k (1 + \log \Delta) / L^p.$

Similarly for any $j > 0$ when we consider the points in $C_i^*$ that have distance between  $2^ja_i^*$ and $2^{j+1} a_i^*$ to the center $\cen_i^*$, the expected cost for all points added before inserting any element to $S$ is at most $\left(\frac{L^p}{k(1 + \log \Delta)}\right).$  After we add a point from this annulus to set $S$, we can use the generalized triangle inequality to bound the cost of any other point $v$ by 
$
\left(2^{p-1}((2^{j+1} a_i^* + d(v,\cen_i^*)^p)\right)\leq
\left(2^{p-1}((2d(v,\cen_i^*))^p + d(v,\cen_i^*)^p)\right).$
Again, the probability of adding this point as a center is 
$\frac{\left(2^{p-1}(2d(v,\cen_i^*))^p + d(v,\cen_i^*)^p\right)k (1 + \log \Delta)}{L^p}.$

Therefore, we can bound the expected number of centers added by:
\begin{align*}
& \sum_{C_i^*}\Bigg(1 +  \log \Delta +\frac{k (1 + \log \Delta)}{L^p}
 \sum_{v\in C_i^*} 2^{p-1}\Big((a^*_i)^p  + (2d(v,\cen_i^*))^p + d(v,\cen_i^*)^p \Big) \Bigg)  \\
& \leq  \sum_{C_i^*} \Bigg(1 + \log \Delta + \frac{k (1 + \log \Delta)}{L^p}
2^{p+1}\sum_{v\in C_i^*} \Big((a^*_i)^p + d(v,\cen_i^*)^p\Big)
\Bigg) \\
& \leq  k(1 +  \log \Delta) + \sum_{C_i^*}\left(k (1 + \log \Delta)2^{p+2}\frac{\cost_p(C_i^*,\centers^p)}{L^p}\right)\\
& = k(1 +  \log \Delta) + \left(k (1 + \log \Delta)2^{p+2}\frac{\opt_p(X)}{L^p}\right)\\
& \leq  \Bigg(1+\frac{2^{p+3}}{\alpha^p}\Bigg)k (1 + \log \Delta).
 \end{align*}

In a similar manner, we can bound the cost by 
\begin{align*}
& \sum_{C_i^*}\Bigg(\frac{L^p}{k(1 + \log \Delta)}(1 + \log \Delta)  +
    \sum_{v\in C_i^*}2^{p-1}\Big((a^*_i)^p + (2d(v,\cen_i^*))^p  + d(v,\cen_i^*)^p\Big)\Bigg)\\
& \le \sum_{C_i^*}\left(\frac{L^p}{k}+ 2^{p+1}\sum_{v\in C_i}\left((a^*_i)^p + d(v,     
    \cen_i^*)^p\right)\right) \\
& \leq (L^p+ 2^{p+2}\opt_p(X)) \leq 2^{p+3}\opt_p(X) 
\end{align*}
Thus the claim follows.
\end{proof}

First step in proving Lemma~\ref{lemma:corollary-of-lemma5o1} using Lemma~\ref{lem:meyerson} is to 
show how to get our guarantees with probability $1-\gamma$. This is done in Algorithm~\ref{alg:full_meyerson}. Lemma~\ref{lem:full_meyerson} states the properties of this algorithm.

\begin{footnotesize}
\begin{algorithm}[h!]
\small
\caption{\small $ComputeMeyerson(X, L^p,\alpha, \gamma, \Delta, p, k )$}\label{alg:full_meyerson}
\begin{algorithmic}[1]
\STATE {\bf Input:} A sequence of points $X$, a lower bound to the optimum $L^p$, a constant $\alpha$ such a that $L_p \ge \alpha \opt_p$, a constant $\gamma$, an upper bound $\Delta$ on the max distance, the parameters $p$ and $k$ of the problem.
 \STATE {\bf Output:} A mapping $\mu: X \rightarrow M$, and a weighted instance $(M, \weight)$ that is consistent with $X$ and cost of the moving $\costm$
\FOR{$i\in[2\log  \gamma^{-1}]$}{
	\STATE $M_i \leftarrow x_0$
	\STATE $\costm_i=0$
	}
\ENDFOR
\FOR{$x\in X_t$}{
	\FOR{$i\in[2\log  \gamma^{-1}]$}{
		\IF{$|M_i|\leq 4 k (1+\log \Delta) \left(\frac{2^{p+3}}{\alpha^p}+1\right)$}{
	                   \IF{$M_i == \emptyset$}
	 	                    \STATE $M_i \leftarrow \{x\}$
	                  \ELSE
      	                             \STATE With probability $\min\left(\frac{k (1 + \log \Delta)d(x,S)^p}{L^p},1\right)$ add $x$ to $M_i$ and set $\weight_i(x)=1$
	                             \STATE Otherwise, let $z \leftarrow \argmin_{y\in S}d(x,y)$, set $\weight_i(z) \leftarrow \weight_i(z)+1$, $\costm_i \leftarrow \costm_i + d(x,y)^p$
	                  \ENDIF	
		}
	        \ENDIF
	}
	\ENDFOR
	}
\ENDFOR
\STATE Let $j$ be the index of the Meyerson sketch of minimum cost $\costm_j$ such that $|M_j|\leq 4 k (1+\log \Delta) \left(\frac{2^{p+3}}{\alpha^p}+1\right)$, if such $j$ does not exist
return $M=\cup_{i=1}^{2\log \gamma^{-1}}M_i$, $\weight_1$, $\infty$
\STATE Extend $\weight_j$ to give weight $0$ to all the points in $M=\cup_{i=1}^{2\log \gamma^{-1}}M_i$ not contained in $M_j$
\STATE Return $M=\cup_{i=1}^{2\log \gamma^{-1}}M_i$, $\weight_j$, $\costm_j$
\end{algorithmic}
\end{algorithm}
\end{footnotesize}

\begin{lemma}\label{lem:full_meyerson}
For a constant $\alpha \in (0, 1)$ and input $X$, Algorithm~\ref{alg:full_meyerson} computes a mapping $\mu$ and a consistent weighted instance $M=\cup_{i=1}^{2\log\gamma^{-1}} (M_i, \weight_i)$ such that with probability at least $1-\gamma$, we have:
\begin{align*}
& |M| \leq 8 k \log  \gamma^{-1}(1+\log \Delta) \left(\frac{2^{p+3}}{\alpha^p}+1\right)\in O(k \log \Delta \log  \gamma^{-1}) \\
& {\text{and}} \qquad \cost_p(X, M)\leq 2^{p+5}\opt_p(X)
\end{align*}
   Furthermore, in the consistent mapping $\mu$ between $X$ and $M$ every point in $X$ is mapped to a point in $M$ inserted before it into the stream. The algorithm computes the cost of the mapping $\costm$.
\end{lemma}

\begin{proof}
As mentioned above,  Lemma~\ref{lem:meyerson} implies that if we construct $2 \log  \gamma^{-1}$ single Meyerson sketches in parallel, with probability in $1-\gamma$, at least one of them gives a constant approximation to the optimum at every point in time and furthermore the single Meyerson contains only $4 k (1+\log \Delta) \left(\frac{2^{p+3}}{\alpha^p}+1\right)$ points.

Now in Algorithm~\ref{alg:full_meyerson} we are almost building $2\log \gamma^{-1}$ single Meyerson sketches, the only difference is that we stop adding points to a single sketch if that becomes too large.  This modification does not change the probability that there exist at least one single sketch that gives a constant approximation to the optimum at every point in time and that contains only $4 k (1+\log \Delta) \left(\frac{2^{p+3}}{\alpha^p}+1\right)$ points.

Thus with probability $1-\gamma$ at least one of the sketches constructed in~\ref{lem:meyerson} gives a constant approximation to the optimum at every point in time. Merging other sketches to this sketch does not affect this property. Furthermore the number of points in each sketch is explicitly bounded by $4 k (1+\log \Delta) \left(\frac{2^{p+3}}{\alpha^p}+1\right)$ so the total number of points in $M$ is bounded by $8 k \log  \gamma^{-1}(1+\log \Delta) \left(\frac{2^{p+3}}{\alpha^p}+1\right)$.

Now we have to prove the existence of a consistent mapping, note that for this we can just use the weighting of a Meyerson sketch with less $4 k (1+\log \Delta) \left(\frac{2^{p+3}}{\alpha^p}+1\right)$ and cost less $2^{p+5}\opt_p(X)$ and assign the weight to all the other points equal to $0$.
\end{proof}

Finally in order to complete the proof of  Lemma~\ref{lemma:corollary-of-lemma5o1}, %~\ref{lem:meyerson},
we need to remove the assumption of knowing a good lower bound for the optimal solution. We do so by trying different guess of the $OPT$ in particular we try for different $L^p$ in $\{m, 2m, 4m, 8m,\dots,$ $2^{\lceil\log M/m\rceil}m\}$. This guarantees that we also run it for a guess $L^p$ such that $L^p\geq \frac{1}{2} \opt_p(X)$. Now we can detect the correct guess by checking the cost of the Meyerson sketches. The pseudo-code of the final algorithm is presented in Algorithm~\ref{alg:guess_meyerson}.

\begin{footnotesize}
\begin{algorithm}[h!]
\small
\caption{\small $SimpleMeyerson(X, m, M, \gamma, \Delta, p, k)$}\label{alg:guess_meyerson}
\begin{algorithmic}[1]
\STATE {\bf Input:} A sequence of points $X$, lower bound $m$ and upper bound $M$ to the optimum, and $\gamma$, an upper bound $\Delta$ on the max distance, the parameters $p$ and $k$ of the problem.
\STATE {\bf Output:} A mapping $\mu: X \rightarrow X'$, and a weighted instance $(X', \weight)$ that is consistent with $X$ with cost of the moving $\costm$
\FOR{$L^p\in \{m, 2m, 4m, 8m,\dots,2^{\lceil\log M/m\rceil}m\}$}{
	\STATE In parallel $ComputeMeyerson(X, L^p,\alpha=1/2, \gamma, \Delta, p, k)$
}
\ENDFOR
\STATE Let $\ell$ be the smallest index (if it exists, otherwise it is an arbitrary index) in $\{m, 2m, 4m, 8m,\dots,2^{\lceil\log M/m\rceil}m\}$ for which the output $ComputeMeyerson$ called with $L_p = \ell$ has size smaller than $8 k \log  \gamma^{-1}(1+\log \Delta) \left(2^{2p+3}+1\right)$ and cost smaller than $2^{p+6}\ell$
\STATE Return the result of the index $\ell$ call for $ComputeMeyerson$.
\end{algorithmic}
\end{algorithm}
\end{footnotesize}

We are now ready to prove the main Lemma of this section.

\begin{proof}[Proof of Lemma~\ref{lemma:corollary-of-lemma5o1}]
The bound on the size of the set and cost follow from the check in Algorithm~\ref{alg:guess_meyerson}. The probability of success is larger than or equal to the probability of success for the $L^p\in \{m, 2m, 4m, 8m,\dots,$ $2^{\lceil\log M/m\rceil}m\}$ such that $L^p\geq \frac{1}{2} \opt_p(X)$. The total space is a result of running $O(\log M/m)$ times $ComputeMeyerson$ calls in parallel and the result in Lemma~\ref{lem:full_meyerson}.
\end{proof}

%% file: app-aug.tex
\section{Augmented Meyerson Sketches}
\label{app:AugmentedMeyerson}
We now show how we can augment the sketch data structure with additional information for clustering in the sliding window model. \\
\subsection{Maintaining Weights.}\label{subsec:MaintainWeight}
As we mentioned above, one of the challenges in adapting the Meyerson sketch to the sliding window setting lies in tracking the weight of a center as points expire. In fact, some of the points mapped initially to a center by $\mu$ may no longer be  part of the sliding window. Formally, given a stream $X = x_1, x_2, \ldots, x_t$, we would like to maintain an estimate of the weight of each center in $\centers$ restricted to the sliding window, i.e., for $\cen \in \centers$, we want to estimate $| \{ x \in W : \mu(x) = \cen\}|$. We denote this quantity by $\weight_\mu(\cen, X_{[t-w, t]})$. 
Maintaining an estimate of such a weight function falls directly into the Smooth Histograms framework introduced by \cite{Braverman07}. We leverage their approach to maintain a $(1 + \epsilon)$ estimate using only $\log _{1+\epsilon} w$ additional overhead per cluster center.
\begin{lemma}\label{lem:estimate-weight}
Fix constant $\epsilon>0$. Using additional space $O\left(|\centers| \log_{1+\epsilon}(w)\right)$, we can extend the Meyerson sketch $\centers, \weight$ with a function $\widehat{\weight} : \centers \times [n] \to \mathbb{Z}$ such that, at every time $t$, for every $ \cen \in \centers$, and every time $\tau$ in the active window: $\weight_\mu(\cen, X_{[\tau, t]}) \leq \widehat{\weight}(\cen,\tau) \leq (1+\epsilon) \weight_\mu(\cen, X_{[\tau, t]})$. 
\end{lemma}

\begin{proof}

In the sketch, for each point $v\in C$, we maintain a sequence of weights $R_v= (r_{v,1}, r_{v,2},\ldots)$ corresponding to the number of points assigned to $v$ by $\mu$ and a sequence of times $T_v = (t_{v,1}, t_{v,2},\ldots)$. The two sequences are initialized as empty.
We preserve the invariant that for each time $t_{v,j} \in T_v$, the number of points assigned to center $v$ by $\mu$ from time $t_{v,j}$ (inclusive) to the end of the stream is equal to $r_{v,j}$.
To do so, when at time $i$, a point is assigned to $v$ by $\mu$ (recall that the assignment is fixed),  we increase by one all of the weights stored in $R_v$ and then we add a new weight initialized to $1$ to $R_v$ and a new time equal to $i$ to $T_v$.

To reduce the size of the structure, we maintain only the significant changes in weights in $R_v$. More precisely, at any time, we delete $r_{v,l}$ and $t_{v,l}$ for any $l \in [2, \ldots, |R_v|-1]$, if $r_{v,l-1} \leq (1+\epsilon)r_{v,l+1}$. We also renumber the indices to be consecutive. Finally, we  remove the $r_{v,l}$ (and the corresponding $t_{r,l}$) for which the weights are larger than $(1+\epsilon)|W|$.\footnote{Note that we can do this because no point has weight larger than $W$ inside the sliding window at any point in the algorithm.}
Notice that at any time, for each $v$ and $l$, either $r_{v,l} = r_{v,l+1} + 1$ or $r_{v,l} \leq (1+\epsilon)r_{v,l+1}$. In fact, if $r_{v,l}$ and $r_{v,l+1}$ refer to consecutive assignments of points to the center, the first case is true. If $r_{v,l}$ and $r_{v,l+1}$   became consecutive after the removal of a point between them, the latter condition is true at the time of the removal, and is preserved by adding $1$ to both elements.

Now, in order to answer $\widehat{\weight}(\cen,\tau)$ we return the  value $r_{i_v}$ in the corresponding $R_v$ array,  where $i_v$ is the index of the largest value smaller or equal to $\tau$ in $T_v$.  Then $\weight_\mu(\cen, X_{[\tau, t]}) \leq \widehat{\weight}(\cen,\tau) \leq (1+\epsilon)\weight_\mu(\cen, X_{[\tau, t]})$.
Finally, note that for each $l \in [|R_v|-2]$,  $r_{v,l} > (1+\epsilon)r_{v,l+2}$, hence, the sequence is decreasing by a factor of $(1+\epsilon)$ every $2$ steps, so that the total length is at most $O(\log_{1+\epsilon}(w))$.
\end{proof}

\subsection{Maintaining Centers.}\label{subsec:MaintainCenter}
Our second task is making sure each cluster has a good center. Once again, this issue is unique to the sliding window setting, as in traditional data streams, points (and thus centers) never expire.  Specifically, whenever a center $\cen$ expires, we aim to replace it with a center $\cen'$ such that, $\dist(\cen, \cen')$ is small. In our context, small means comparable to the distance between the center $\cen$ and any point $x$ in the window that was mapped to it. 

\begin{definition}
Fix a mapping $\mu$ and a center $\cen$. We say that $y \in X_{[\tau, t]}$ is an $\epsilon$-replacement for $\cen$ at time $\tau$ if:
$\dist(\cen, y)^p \leq (1 + \epsilon) \dist(\cen, x)^p$ for any $x \in X_{[\tau, t]}$ with $\mu(x) = \cen$, where $t$ is the last time in the current stream.
\end{definition}

To obtain an $\epsilon$-replacement, we consider a sequence of concentric shells with geometrically increasing radii around each center $\cen$, and keep the last occurring point in every shell. Since distances are between $1$ and $\Delta$, we can bound the total number of shells by $O(\log_{1+\epsilon} \Delta)$ per center. When $\cen$ needs to be replaced, we simply look for the smallest non-empty shell and return its representative. By construction, we will ensure that the $\epsilon$-replacement property is satisfied. 

 \begin{lemma}\label{mapp:m}

Let $\epsilon \in (0,1)$, then with space $O(|\centers| \log_{1+\epsilon}(\Delta))$, it is possible to construct a function $\last : \centers \times [n] \to X$ that will always return an $\epsilon$-replacement for $\cen \in \centers$. 
\end{lemma}

\begin{proof}
Recall that the minimum distance between two distinct points is $1$, and that $\Delta$ is the maximum distance between two points in the dataset. Let $\Lambda$ to be the set $\{1, (1+\epsilon), (1+\epsilon)^2, \ldots \Delta^p\}$ with both extremes included. For each point $y\in \centers$, we maintain a sequence of elements $L_y = \{l_{y,\lambda} | \lambda \in \Lambda)$ and a sequence of times $H_y = \{h_{y,\lambda} | \lambda \in \Lambda\}$ indexed by $\lambda\in\Lambda$. When a point $y$ is added to $\centers$, we initialize the sequence $L_y$ with $y$ and  $H_y$ with its insertion time.  Whenever a point $x$ is mapped by the Meyerson sketch to $y$ at time $i$, let $d=\dist(y, x)^p$, we set $l_{y,\lambda}= x$ and $h_{y,\lambda}= i$ for all $\lambda \geq d$.
First note that we never add a point $x$ with distance larger than $\Delta$ by definition, so each point is added to at least one set.

Now, to answer $\last(y,\tau)$, we return the point $l_{y,\lambda}$ for the smallest $\lambda$ such that $h_{y,\lambda}\ge \tau$. If no such $\lambda$ exists, we return $\emptyset$.

Note that by construction, we return a point in $\{x_{\tau},\dots,$ $x_t\}$ that is mapped to $y$ if and only if the set is non-empty. Now, suppose that we return $y'$ and $y' \in l_{y,\lambda}$. 
\comment{We know that $\dist(x, y)^p \le \lambda'$. If $\lambda' =  \frac{2^{-p}m}{n}$ the result follows from the fact that $\dist^p(x, y) \le \frac{2^{-p}m}{n}$ and that the distances are non-negative. Suppose that $\lambda' > \frac{2^{-p}m}{n}$ then}
We know that there exists no point $x \in \cap \{x_{\tau},\dots,x_t\}$ mapped to $y$ by the Meyerson sketch such that $\dist(x, y)^p < \frac{\lambda}{1+\epsilon}$; otherwise, the point would be stored in $l_{y,\frac{\lambda}{1+\epsilon}}$, which is a contradiction as we returned $l_{y,\lambda}$. Hence, we have $\dist(y, y')^p\leq (1+\epsilon)\dist(x,y)^p$ for any $x \in \cap \{x_{\tau},\dots,x_t\}$ mapped to $y$ by the Meyerson sketch.

Finally, notice that we store $O(\log_{1+\epsilon}\Delta)$ elements for each element in $\centers$ (notice that $p$ is a constant).
\end{proof}

\subsection{Suffix Sketch}
\label{app:suffix-sketch}
The two augmentations before allow us to modify the Meyerson sketch in such a way that it can return an approximate solution 
for any suffix of the stream with length at most $w$. More precisely, we will show that we maintain a valid solution for the suffix 
with a cost comparable to the optimum of {\it the entire stream}.

Before proving the main result of this section, we need to add some additional bookkeeping to estimate the cost of the sketch inside the sliding window. The idea is similar to that used to maintain the approximate weight of a point.

\begin{lemma}\label{lem:estimate-cost}
Let $\epsilon\in (0,1)$ be a constant. Then using additional space $O\left(\log_{1+\epsilon}(M/m)\right)$ we can extend the Meyerson sketch $(\centers, \weight)$ with a function $\widehat{\costm} : [n] \to \mathbb{R}$ such that, for every time $\tau$ in the active window: $\costm(\tau) \leq \widehat{\costm}(\tau) \leq (1+\epsilon) \costm(\tau)$, where $\costm(\tau)$ is the mapping cost for points inserted after time $\tau$ 
\end{lemma}
\begin{proof}
The proof is very similar to the proof of Lemma~\ref{lem:estimate-weight} and we include it here for completeness.

In the sketch, we maintain a sequence of cost $B= (b_1, b_2,\ldots)$ corresponding to different moving costs of $\mu$ and a sequence of times $G = (g_1, g_{2},\ldots)$. The two sequences are initialized as empty.
We preserve the invariant that for each time $g_{j} \in G$, the moving cost of $\mu$ for points inserted after time $g_{j}$ (inclusive) to the end of the stream is equal to $b_{j}$.
To do so, when at time $i$, a point, $x_i$ is assigned to some center $y$ by $\mu$ (recall that the assignment is fixed),  we increase by $d(x_i,y)^p$ all of the weights stored in $B$ and then we add a new weight initialized to $d(x_i,y)^p$ to $B$ and a new time equal to $i$ to $G$.

To reduce the size of the structure, we maintain only the significant changes in weights in $B$. More precisely, at any time, we delete $b_{l}$ and $g_{l}$ for any $l \in [2, \ldots, |B|-1]$, if $b_{l-1} \leq (1+\epsilon)b_{l+1}$. We also renumber the indices to be consecutive. Finally, we  remove the $b_{l}$ (and the corresponding $g_{l}$) for which the weights are larger than $(1+\epsilon)2^{p+7} M$.\footnote{Note that we can do this because mapping of a good Meyerson sketch cost more than $(1+\epsilon)2^{p+7} M$ by Lemma~\ref{lemma:corollary-of-lemma5o1}.}
Notice that at any time, for each $l$, either $b_{l} = b_{l+1} + 1$ or $b_{l} \leq (1+\epsilon)b_{l+1}$. In fact, if $b_{l}$ and $b_{l+1}$ refer to consecutive assignments of points to the center, the first case is true. If $b_{l}$ and $b_{l+1}$   became consecutive after the removal of a point between them, the latter condition is true at the time of the removal, and is preserved by adding $1$ to both elements.

Now, to compute the cost of mapping for every suffix, $\widehat{\costm}(\tau)$, we return the value $b_{i}$ in the $B$ array,  where $i$ is the index of largest value smaller or equal to $\tau$ in $G$.  Then $\costm(\tau) \leq \widehat{\costm}(\tau) \leq (1+\epsilon) \costm(\tau)$.
Finally, note that for each $l \in [|B|-2]$,  $b_{l} > (1+\epsilon)b_{l+2}$, hence, the sequence is decreasing by a factor of $(1+\epsilon)$ every $2$ steps, so that the total length is at most $O(\log_{1+\epsilon}(M/m))$.
\end{proof}
We note that the the algorithm defined previously can be easily modified to add the bookkeeping described in the previous lemmas, in particular we need only to modify the $ComputeMeyerson$ function. 

Now we can formally prove the main properties of our augmented Meyerson sketch.

\begin{footnotesize}
\begin{algorithm}[t]
\small
\caption{$\textrm{AugmentedMeyerson}(X, w, m, M, \Delta)$}\label{alg:augmented_meyerson}
\begin{algorithmic}[1]
\STATE {\bf Input:} A sequence of points $X$, lower bound $m$ and upper bound $M$ to the optimum, and an upper bound $\Delta$ on the max distance.
\STATE {\bf Output:} A data structure allowing to extract a $\epsilon$-consistent weighted instance of the substream $X_{[\tau, t]}$ with centers belonging to $X_{[\tau,t]}$.
\STATE Let parameters $p$ and $k$ be the parameters of the $k$-clustering problem.
\STATE Let $\gamma$ be the desired probability of success of the individual Meyerson algorithm.
\STATE Let $\epsilon$ be the parameter of the $\epsilon$-consistent mapping and $\epsilon$-replacement of centers as described in sections~\ref{subsec:MaintainWeight} and \ref{subsec:MaintainCenter}.
\FOR{$L^p\in \{m, 2m, 4m, 8m,\dots,2^{\lceil\log M/m\rceil}m\}$}{
	\STATE In parallel run $ComputeMeyerson(X, L^p,\alpha=1/2, \gamma, \Delta, p, k)$ with bookkeeping for maintaining $\epsilon$-consistent mappings and $\epsilon$-replacement of centers as described in sections~\ref{subsec:MaintainWeight} and \ref{subsec:MaintainCenter}.
}
\ENDFOR
\STATE Let $\ell$ be the smallest index (if it exists, otherwise it is an arbitrary index) in $\{m, 2m, 4m, 8m,\dots,2^{\lceil\log M/m\rceil}m\}$ for which the output $ComputeMeyerson$ called with $L_p = \ell$ has size smaller than $8 k \log  \gamma^{-1}(1+\log \Delta) \left(2^{2p+3}+1\right)$ and cost smaller than $2^{p+6}\ell$
\STATE Return the result of the index $\ell$ call for $ComputeMeyerson$ with the needed booking for allowing the operation $\suffix_\tau$.
\end{algorithmic}
\end{algorithm}
\end{footnotesize}

\begin{lemma}[Lemma~\ref{lemma:mb} restated]
Let $w$ be the size of the sliding window, $\epsilon\in (0,1)$ be a constant and $t$ the current time. Let $(X,\dist)$ be a metric space and fix $ \gamma \in (0,1)$. The augmented Meyerson algorithm computes
an implicit mapping $\mu: X \rightarrow \centers$, and
 an $\epsilon$-consistent weighted instance $(\centers, \widehat{\weight})$ for all substreams $X_{[\tau, t]}$ with $\tau \geq t - w$ such that,  with probability $1- \gamma$, we have: \\
$$ |\centers| \leq 2^{2p+8} k \log  \gamma^{-1}\log \Delta \qquad {\text{and}} \qquad $$ %\\
$$ \cost_p(X_{[\tau,t]}, \centers)  \le  2^{2p+8} \opt_p(X).$$ %\\
The algorithm uses space $O( k \log  \gamma^{-1} \log \Delta \log M/m (\log M + \log w + \log \Delta))$ and stores the cost of the consistent mapping, $\cost(X, \mu)$, and also a $1+\epsilon$ approximation to the cost of the $\epsilon$-consistent mapping, $\widehat{\cost}(X_{[\tau, t]}, \mu)$.
\end{lemma}

\begin{proof}
The main idea of the proof is to use our three bookkeeping tricks to provide a good sketch also for an arbitrary suffix of a stream. In particular our algorithm will run the Meyerson sketch with the additional bookkeeping as presented in Algorithm~\ref{alg:guess_meyerson} and output a weighted instance and its cost using the additional bookkeeping.

Note that by Lemma~\ref{lemma:corollary-of-lemma5o1} we know that with probability $1-\gamma$ the sets of center that we use in a Meyerson sketch are at most $2^{2p+8} k \log  \gamma^{-1}\log \Delta$, our augmented Meyerson sketch change the center as describe in Lemma~\ref{mapp:m} to guarantee that $\centers \subseteq X_{[\tau,t]}$ but it never increases their number so the final number of center is still bounded by $2^{2p+8} k \log  \gamma^{-1}\log \Delta$.

To bound the cost of the moving, we note that  $\cost_p(X_{[\tau,t]}, \centers)$ is initially bounded by $2^{p+7} \opt_p(X)$. Although in the augmented sketch we change the set of centers so the cost could increase, nevertheless for every point $y\in X_{[\tau,t]}$, we are guaranteed that there is a point in the final set of centers at distance at most $2^{p+1}d(y,\mu(y))^p$. This follows from Lemma~\ref{mapp:m} and triangle inequality and because $2+\epsilon < 4$. Finally the  the factor $2^{2p+8}$ in the statements comes from multiplying $2^{p+1}$ with the approximation factor of Meyerson $2^{p+7}$.

Finally we notice that we can compute $\widehat{\weight}$ using the bookkeeping and the algorithm presented in Lemma~\ref{lem:estimate-weight} and the $\widehat{\costm}(\tau)$ using bookkeeping and the algorithm presented in Lemma~\ref{lem:estimate-cost}.

The space bound follows from the space of the simple Meyerson sketch times the required space for bookkeeping. 
\end{proof}

For completeness, we provide a pseudocode for the $\textrm{AugmentedMeyerson}$ (Algorithm~\ref{alg:augmented_meyerson})  which is used in Algorithm~\ref{alg:pseudo}. The $\textrm{AugmentedMeyerson}$ corresponds to $\textrm{SimpleMeyerson}$ (Algorithm~\ref{alg:guess_meyerson}) with the additional bookkeeping described in Section~\ref{app:AugmentedMeyerson}. More precisely, $\textrm{AugmentedMeyerson}$ stores for each sketch the data structures need for maintaining the estimate of the weights and centers described in sections~\ref{subsec:MaintainWeight} and \ref{subsec:MaintainCenter}.
Using these data structure the algorithm we can extract the weighted set $\suffix_\tau(S)=(\centers, \widehat{\weight})$ for a given $\tau$ which represents an $\epsilon$-consistent weighted instance of the substream $X_{[\tau, t]}$ with centers $C$ belonging to $X_{[\tau,t]}$.

%% file: app-sw.tex
\section{Omitted Proofs from Sliding Windows subsection}
\label{app:sliding-window}

\begin{lemma}[Lemma~\ref{lem:aas} restated]
Using an approximation algorithm $\algo$, from the augmented Meyerson sketch $S(Z)$, with probability $\ge 1-\gamma$, we can output a solution $\algo(S(Z))$ and an estimate  $\acost_p(S(Z), \algo(S(Z)))$ of its cost s.t. 
\begin{align*}
\cost_p(Z, \algo(S(Z)))  & \le  \acost_p(S(Z), \algo(S(Z))) \\
 & \le \beta(\rho) \cost_p(Z, \opt(Z))
\end{align*}
for a constant $\beta(\rho) \le 2^{3p+6}\rho$ depending only the approximation factor $\rho$ of $\algo$.
\end{lemma}

\begin{proof}
Given the sketch $S(Z)$ computed with Algorithm~\ref{alg:guess_meyerson}, we obtain a set of centers $Y(Z)$ and theirs weights $\weight(Z)$, and a cost of the mapping $\costm(Z)$ (as well as an implicitly defined consistent mapping $\mu$). 

So we can construct a weighted instance $(Y(Z), \weight(Z))$ and solve it approximately with $\algo$ to obtain centers $\centers= \algo(S(Z))$ which are given in output.  We now bound the cost of $\centers$ over $Z$. 

As we have shown in the proof of Lemma~\ref{lem:cmap}, we have that $\cost_p(Z, \centers) \le 2^{p-1} (\costm (Z) + \cost_p(Y(Z), \weight(Z),  \centers))$, where $\cost_p(Y(Z), \weight(Z),  \centers)$ is the cost of the solution $\centers$ over the weighted instance. We also know that the weighted instance has an optimum cost of $\le 2^{2p-1} (\costm (Z) +  \opt_p(Z))$. So applying an $\rho$-approximate algorithm we get 
$$\cost_p(Y(Z), \weight(Z),  \centers) \le \rho 2^{2p-1} (\costm (Z) +  \opt_p(Z)) $$

Finally we get that  $\cost_p(Z, \centers) \le \rho  2^{2p} (\costm (Z) + \opt_p(Z))$ which by Lemma~\ref{lem:full_meyerson} gives $\cost_p(Z, \centers) \le 2^{3p+6} \rho \opt_p(Z)$, with probability $\ge 1-\gamma$, so the theorem holds with $\beta(\rho) := 2^{3p+6} \rho$.

Finally, notice that the value $\acost_p(S(Z), \algo(S(Z)))  := 2^{p-1} (\costm (Z) + \cost_p(Y(Z), \weight(Z),  \centers))$, can be computed from $S(Z)$ (using $\algo$) so we can output it as an estimate of the cost. Notice, that this is valid a lowerbound of the actual cost on $Z$ of the solution, and it is $ \le \beta (\rho) \opt_p(Z)$. 
 \end{proof}
 
 Now we prove some basic properties of Algorithm~\ref{alg:pseudo}] that we will use in other prove.
 
  \begin{lemma}[Invariants of Algorithm~\ref{alg:pseudo}]\label{lem:invariantsofAlg}
For a set $A_{\lambda}$, let $A^+_\lambda$ be the set $A_\lambda$ together with the first element of the corresponding $B_{\lambda}$. The following are invariants  maintained by the algorithm: 
(i) $A_{\lambda}$ and $B_{\lambda}$ are two disjoint consecutive substreams of $X$,
(ii) $A_{\lambda}$ precedes $B_{\lambda}$, and 
(iii) $B_{\lambda}$ ends with the current time $t$ and always contains the last element of the stream.
(iv) When $|A_\lambda| \geq 1$: %\not= \emptyset$:
$\acost_p(S(A_{\lambda}^+),\algo (S(A_{\lambda}^+))) > \lambda.$
(v) When $|A_\lambda| \geq 2$: %\not= \emptyset$:
$\acost_p(S(A_{\lambda}), 
\algo(S(A_{\lambda}))) \leq \lambda $.
(vi) When $|B_\lambda| \geq 2$:
 $\acost_p(S(B_{\lambda}),\algo(S(B_{\lambda}))) \leq \lambda$
 \end{lemma}
\begin{proof}
The first three points follow from construction. The fourth follows from the fact that $A_\lambda$ is set to a non-empty sketch only when the value associated with $B_\lambda \cup \{ x \}$ is$> \lambda$ and this property is maintained thereafter. Equation (v) follows from the fact that $|A_\lambda|\ge 2$ implies that $A_\lambda$ has been set as a copy of a prior $B_\lambda$ which was not a singleton, and this happens only if the value of the cost of the sketch associated with $B_\lambda$ is $<\lambda$. Equation (vi) follows from the same property.
\end{proof}

\begin{lemma}[Composition with a Suffix of stream, restatement of Lemma~\ref{lem:suffix-composition}] 
 Given two substreams $A$,$B$ (with possibly $B=\emptyset$) and a time $\tau$ in $A$, let $\algo$ be a constant approximation algorithm for the $k$-clustering problem. Then if $\opt_p(A) \le O(\opt_p(A_\tau \cup B)$, then, with probability $\ge 1-O(\gamma)$, we have  $f_p(A_\tau \cup B, \algo(\suffix_\tau(S(A)) \cup S(B))) \le O(\opt_p(A_\tau \cup B)$.
\end{lemma} 
\begin{proof}
Let $\mu_A$ be the mapping consistent with the sketch $S(A)$ and $\mu_B$ be the mapping consistent with the sketch $S(B)$. Note that from this two mappings we can construct a third mapping $\mu$ such that $\mu(x)=\mu_A(x)$ if $x\in A$ and $\mu(x)=\mu_B(x)$ if $x\in B$(note that input points are in general position so they are all in distinct position). The mapping $\mu$ is now a valid mapping for $A_\tau \cup B$, furthermore we can note that the number of points mapped to a point $y$ by the mapping is equal to $|\{x\in A: \mu(x)=y\}|+|\{x\in B: \mu(x)=y\}|$, so by definition of the sketch $\suffix_\tau(S(A)) \cup S(B)$ and by the fact that $\suffix_\tau(S(A))$ is $\epsilon$-consistent and $S(B)$ is consistent we obtain that $\suffix_\tau(S(A)) \cup S(B)$ is $\epsilon$-consistent with $\mu$.

Note from Lemma~\ref{lemma:corollary-of-lemma5o1} we know that the moving cost of $\mu_A$ on $A$ is in $O(\opt_p(A))$ with probability $\ge 1-\gamma$. Similarly the moving cost of $\mu_B$ on $B$ is in $O(\opt_p(B))$ with probability $\ge 1-\gamma$, so the moving cost of the function $\mu$ on $A\cup B$ is $O(\opt_p(A)+\opt_p(B)) \in O(\opt_p(A_\tau \cup B)$ by the hypothesis in our Lemma with probability $\ge 1-2\gamma$. Furthermore also the moving cost of $\mu$ on $A_\tau \cup B$ is in $O(\opt_p(A_\tau \cup B))$. 

Now, let $\centers_\algo$ be the centers selected by a constant approximation algorithm on $\suffix_\tau(S(A)) \cup S(B)$. From generalized triangle inequality we get 
$\cost_p(A_\tau \cup B, \centers_\algo) =$ $\sum_{x\in A_\tau \cup B}$ \\*
$d(x, \centers_\algo)^p$ $\leq \sum_{x\in A_\tau \cup B} 2^{p-1}(d(x, \mu(x))^p + d(\mu(x), \centers_\algo)^p)\in O(\opt_p(A_\tau \cup B) + \sum_{x\in A_\tau \cup B} d(\mu(x), \centers_\algo)^p)$. Now from Lemma~\ref{lem:cmap} we know that the instance obtained by applying the mapping $\mu$ to $A_\tau \cup B$ has optimal solution $\centers_\mu$ with cost in $O(\opt_p(A_\tau \cup B))$. Unfortunately we do not apply the algorithm directly on an instance consistent with the mapping, nevertheless by $\epsilon$-consistency we know that the cost of the set of centers $\centers_\mu$ for the instance $\suffix_\tau(S(A)) \cup S(B)$ is most $(1+\epsilon)$ times the cost for the instance obtained by the mapping. So also the optimal solution of $\suffix_\tau(S(A)) \cup S(B)$ has cost in $O(\opt_p(A_\tau \cup B))$. So also $\centers_\algo$ have cost in $O(\opt_p(A_\tau \cup B))$ because $\algo$ is a constant approximation algorithm. The Lemma follows by noticing that by definition $\epsilon$-consistent $\sum_{x\in A_\tau \cup B} d(\mu(x), \centers_\algo)^p \leq \cost_p(\suffix_\tau(S(A)) \cup S(B), \centers_\algo)$.
\end{proof}

\begin{theorem}[Theorem~\ref{th:algo-main-th} restated]
With probability $1-\gamma$, Algorithm~\ref{alg:algo-min}, outputs an $O(1)$-approximation for the sliding window $k$-clustering problem using space:
$O\big(k \log (\Delta) (\log (\Delta) + \log(w) + \log(M) )$
$\log^2 (M / m) \log(\gamma^{-1} \log(M/m)) \big)$
and total update time 
$O(  T(k \log (\Delta), k)$ $\log^2 (M / m) \log(\gamma^{-1}  \log(M / m))$ $ (\log (\Delta) + \log(w) + \log(M) )$.
\end{theorem}
\begin{proof} 
Notice that at any point in time we run at most $2 \log_{1+\delta} (M/m)$ Augmented Meyerson sketches. We can set each of them to use as probability of error bound $\frac{\gamma}{2 \log_{1+\delta} (M/m)}$ to have total probability of any of them failing $\le 1- \gamma$. We continue the analysis on assuming we are in the case all current sketches did not fail.

Consider the sketches maintained by the algorithm. 
If for some $\lambda$, the active window $\activewindow$ is identical to the interval $B_{\lambda}$, we return $\algo(S_2)$ as the solution. Since $S_2$ is a Meyerson sketch we have that the weighted instance computed by the sketch is consistent with a mapping of cost at most a constant factor larger than the cost of the optimal solution in the active window and so by Lemma~\ref{lem:cmap} we have that $\algo(S_2)$ is a constant approximation. (Note that this is independent of the value of $\lambda$, in fact $\lambda$ is not used by the Meyerson sketch that independently tries several possible lower bounds for the cost of the solution)%(by composing $B_{\lambda}$ with the empty sketch).

Otherwise, we find the maximum $\lambda' \in \Lambda$ for which $A_{\lambda'}$ is not empty and fully contained in the current active window.  As long as $|W| > 1$, for the sketches associated with the smallest $\lambda$, we have $A_m \neq \emptyset$, which guarantees that such a $\lambda'$ exists. Here $m$ is the lower bound on the optimum value as defined in Section~\ref{sec:preliminaries}. Furthermore, let $\lambda^*  = \lambda'  (1+\delta)$. Note that the fact that for the largest $\lambda$, all elements are contained in $B_{\lambda}$, this guarantees that such a $\lambda^*$ exists in the set of thresholds $\Lambda$  as well. % (we can't have an estimated cost $>2^p \beta M$ in any substream). 

We now show that $\lambda'  < O(\opt(W))$, which implies that $\lambda^* < (1+\delta) O( \opt(W))$. 
In fact, suppose $\lambda'$ is such that $A_{\lambda'}$ is non-empty. We know that $A_{\lambda'} \subseteq W$. Let $A^+_{\lambda'}$ be $A_{\lambda'}$ plus the first element of $B_\lambda$, by definition of the algorithm, we have that 
$\lambda'<  \acost_p(S_{A^+_{\lambda'}},\algo(S(A_{{\lambda'}}^+))) $
 where the inequality follows from the invariant (iv) of Lemma~\ref{lem:invariantsofAlg}, but now $\acost_p(S_{A^+_{\lambda'}},\algo(S(A_{{\lambda'}}^+)))$ is the cost of a solution computed on the Meyerson sketch for $A^+_{\lambda'}$ for $A^+_{\lambda'}$ so using the fact that the Meyerson sketch provides a constant approximation we get 
 $ \lambda'<  \acost_p(S_{A^+_{\lambda'}},\algo(S(A_{{\lambda'}}^+))) \le O( \opt(A^+_{\lambda'}))$
 but now we can notice that by adding points to any set $X$ we can decrease the cost of the solution by at most a factor $2^p$(for a formal proof of this fact look at Lemma~\ref{lem:adding}) so we get 
$  \lambda'<  \acost_p(S_{A^+_{\lambda'}},\algo(S(A_{{\lambda'}}^+))) \le O( \opt(A^+_{\lambda'})) \le 2^pO( \opt(W))\in O(\opt(W)).
$

So, we can then assume that $\lambda^* < (1+\delta)O( \opt(W))$. By definition of $\lambda^*$, $W \subseteq A_{\lambda*} \cup B_{\lambda*}$. There are three cases to consider. 
\begin{itemize}\itemsep=0in\vspace{-0.15in}
\item $\activewindow$ is a strict suffix of $B_{\lambda^*}$. In this case, the $\activewindow$ is equal to a suffix of $B_{\lambda^*}$ starting at position $\tau$ for some $\tau$. Notice that $|B_{\lambda^*}| \ge 2$ so we know that the sketch over $B_{\lambda^*}$ has cost $\le \lambda^*$. so we can apply the Lemma~\ref{lem:suffix-composition} to compose  the suffix of $B_{\lambda^*}$ with the empty sketch. 
\item $\activewindow = B_{\lambda^*}$, this ensures to a constant approximation. 
\item $\activewindow$  intersects with $A_{\lambda^*}$. We start by computing $\suffix_t(S_1)$, and then we compute $\algo(\suffix_t(S_1)\cup S_2)$. Again we know that the $A_{\lambda^*}$ has a sketch of small cost so we can apply the composition Lemma~\ref{lem:suffix-composition}.
\vspace{-0.1in}
\end{itemize}
Notice that the space bound comes from the total number of points in all sketches used. Notice also that the number of distance function evaluations (and thus the update time) to update the Meyerson sketches per each point in the stream is bounded by the number of points in all sketches (we do one evaluation for each center). The total update time depends also on the complexity of $\algo$, but we only run it on instances of $k \log (\Delta)$ points, so the total update time is 
$
 O\big(  T(k \log (\Delta), k) \log^2 (M / m) \log(\gamma^{-1}  \log(M / m))  (\log (\Delta) + \log(w) + \log(M) \big).
$
Notice also that at the last step, to return a solution, after the updates are done, we only require solving a constant number of $k$-clustering instances with $\algo$, again each of them of size $O(k \log (\Delta))$.
\end{proof}

%% file: app-exp.tex
\section{Empirical analysis}
 \label{sect:app-exp}

\subsection{Additional information on the experimental setup}
\textbf{Baselines.} We consider the following baselines.
\textbf{Off-Line K-Means++}:  We use $k$-means++ over the entire window as a proxy for the optimum, since the latter is NP-hard to compute. At every insertion, we report the best solution over 10 runs of $k$-means++ on the window.  Observe that this is inefficient as it requires $\Omega(w)$ space and $\Omega(kw)$ cost per update. 
 Specifically, we run $k$-means++ 10 times and report the best solution obtained. To compare the cost of our algorithm to the optimum, we run $k$-means++ every time a new element is inserted. While this gives the best cost, this is obviously inefficient: such an approach requires $\Omega(w)$ space and $\Omega(kw)$ cost per update.  
 \textbf{Sampling}: We maintain a random sample of points from the active window, and then run $k$-means++ on the sample. This allows us to evaluate the performance of a baseline, at the same space cost of our algorithm. 
\textbf{SODA16}: We also evaluated the only previously published algorithm for this setting  in~\cite{DBLP:conf/soda/BravermanLLM16}. 
We note that we made some practical modifications to further improve the performance of our algorithm which we  now describe.

For our algorithm we used a standard optimization from the literature \cite{lattanzi2017consistent}: running one copy of the Meyerson sketch instead of the $O(\log(1/\gamma))$ copies that are needed for high probability statements. We also developed a lazy evaluation of the cost of the Meyerson sketch that saves update time.
For the constant $\alpha$ in the Meyerson algorithm we use $.5$, as shown in the pseudo-code.

Second, as we focus on the cost of the solution, instead of the centers output, we ignore the presence of the center in the window (i.e., we do not use the method to output centers in the window). This is consistent with the problem addressed by~\cite{DBLP:conf/soda/BravermanLLM16} in which they do not restrict the centers to be in the sliding window. This way we allow a fair comparison of our results with this baseline.

Third, similarly to~\cite{lattanzi2017consistent}, in the implementation of the update function, the cost of the sketch is estimated through a lazy evaluation in which the cost is not recomputed from the last time unless at least one new center is added or the number of points or total cost associated with a center is increased by a factor $(1+\eta)$ from the last full evaluation (we use $\eta = 0.05$); we stop a Meyerson sketch only when the size bound is exceeded, not the cost bound. 

Finally, to output a solution, instead of the specific pair of summaries described in our algorithm, we make best-effort use of all available pairs of sketches associated with all guesses of the optimum and return the best one with lowest estimated cost--this also makes the algorithm more resistant to error in estimating the lower and upper bounds of the cost. Our empirical analysis shows that these techniques speed up the computation significantly while retaining very strong approximation factors in real dataset, as we show in this section.

{\bf Upper and lower bounds on cost.} 
Our algorithms described before require an estimate of the lower and upper bounds for the cost to initialize the thresholds used in the summaries. In the Appendix~\ref{app:assumption} we show that such assumptions can be removed at the expenses of a more complicated algorithm. In our experiments instead of removing the assumptions we have implemented the simpler algorithms described in the paper and we used an heuristic to estimate the bounds $m,M$ in input to the algorithms. Here we describe the heuristic, which empirically results in strong approximation results and efficiency: sample a number of window-sized substreams from (a prefix) of the stream and compute an approximate solution for each window with any algorithm (e.g., kmeans++ or an insertion-only stream algorithm); then use the min and max cost found (with a slack) for the bounds. 

More precisely, let $m',M'$ be the minimum and maximum cost observed in the samples and let $\mu, \sigma$ be the empirical mean and standard deviation of the samples' costs. To estimate $m,M$ we use $m := \max(\frac{m'}{3},\mu-3\sigma)$ and $M := \max(3M,\mu+3\sigma)$. We use $10$ samples in our experiments for the estimation.  Note that the $\max$ with $\frac{m'}{3}$ is used to ensure that the lower bound is set to be $>0$. We use the same technique in the evaluation of the SODA16 baseline~\cite{DBLP:conf/soda/BravermanLLM16}. For the $\Delta$ input of our algorithms we use the $M'$ obtained from the heuristic.

{\bf Implementation details for \cite{DBLP:conf/soda/BravermanLLM16}}.
To the best of our knowledge, this algorithm has not been evaluated empirically. In implementing this algorithm we used the same speed up techniques used for our algorithm: we used one instead of logarithmic many Meyerson sketches, we used the same implementation of the Meyerson sketch used for our algorithm for fairness of comparison. We also provide the same upper and lower bounds on the optimum used for our algorithm. Finally, we set the grid width parameter $\delta$ of the algorithm, as the equivalent $\delta$ parameter of our algorithm. We set similarly the $\beta$ parameter~\cite{DBLP:conf/soda/BravermanLLM16}.

\subsection{More detailed comparison with previous work}

\begin{table}
\small
\centering
\begin{tabular}{ c | c | c | c | c |}
\textbf{Dataset} & $k$ &
\makecell{\textbf{Space Decr.} \\ \textbf{Factor}} &
\makecell{\textbf{Speed-Up} \\ \textbf{Factor}}  &
\makecell{\textbf{Cost} \\ \textbf{(ratio)}}  \\
\hline
\multirow{5}{*}{COVER}&4&5.23&	10.88&	$99.5\%$ \\
 &8&4.86&	10.30&	$99.9\%$\\
&10&4.86&	10.36&	$97.6\%$\\
 &16&6.04&	13.10&	$95.1\%$\\
&20&7.09&	16.45&	$94.3\%$\\
\hline
\multirow{3}{*}{SHUTTLE}&4&5.07&	9.09&	$106.8\%$\\
&8&6.09&	7.89&	$102.4\%$\\
&10&6.28&	7.79&	$103.2\%$\\
&16&15.32&	15.14&	$118.9\%$\\
&20&13.30&	12.32&	$143.7\%$\\
\end{tabular}
\caption{Decrease of space use, decrease in (speed-up) and ratio of mean cost of the solutions of our algorithm vs the SODA16 baseline. Space decrease and the speed-up factors are multiplicative, the cost column reports the ratio of the cost of our algorithm over that of the SODA16 algorithm.
}
\label{table:soda16-extended}
\vspace{-0.15in}
\end{table}

Here we report a more detailed analysis in Table~\ref{table:soda16-extended} where we show that our algorithm uses up to $15x$ less space, does up to $16x$ fewer distance calculations and achieves a solution up to 7\% cheaper. Notice how the improvements over the baseline grow with $k$ as predicted by theory.  We observe that for $k=16$ on the SODA16 baseline always has roughly the same space or more that that required to store the entire sliding window making the sliding window algorithm trivial.

Notice how our algorithm result in substantial improvements even for small values of $k$. 

\subsection{Additional experiments on cost of the solution}

\begin{table}
\small
\centering
\begin{tabular}{ c | c | c | c | c |}
\textbf{Dataset} & $W$ & \textbf{Sketch} & \textbf{Sampling}  & \textbf{Reduction}  \\
\hline
\multirow{3}{*}{COVER}&10000&1.09 (0.14)&1.29 (0.31)&$-15.9\%$ \\
&20000&1.05 (0.12) &1.31 (0.30)&$-19.4\%$\\
&40000&1.06 (0.13) &1.41 (0.44)&$-25.0\%$\\
\hline
\multirow{3}{*}{SHUTTLE}&10000&1.11 (0.16)&1.48 (0.25)&$-24.8\%$\\
&20000&1.09  (0.14)&1.32 (0.14)&$-17.4\%$\\
&40000&1.07 (0.11)&1.26 (0.13)&$-15.0\%$\\
\hline
\multirow{3}{*}{SKIN}&10000&1.13 (0.19)&1.40 (0.40)&$-19.7\%$\\
&20000&1.10 (0.17)&1.34 (0.32)&$-17.8\%$\\
&40000&1.09 (0.14) &1.31 (0.28)&$-16.4\%$\\
\hline
\end{tabular}
\caption{Mean cost ratio of our algorithm (std-dev in parenthesis) and of the sampling baseline over the $k$-means++ gold standard for $k=10$, $\delta=0.2$}
\label{table:cost-over-baseline}
\end{table}

\begin{figure}
\begin{center}
\subfigure[COVERTYPE]{\includegraphics[width=0.3\textwidth,keepaspectratio]{figs/covtype-40000-20-0_2-cost.pdf}\label{fig:evolution-covtype-ext}}
\subfigure[SHUTTLE]{\includegraphics[width=0.3\textwidth,keepaspectratio]{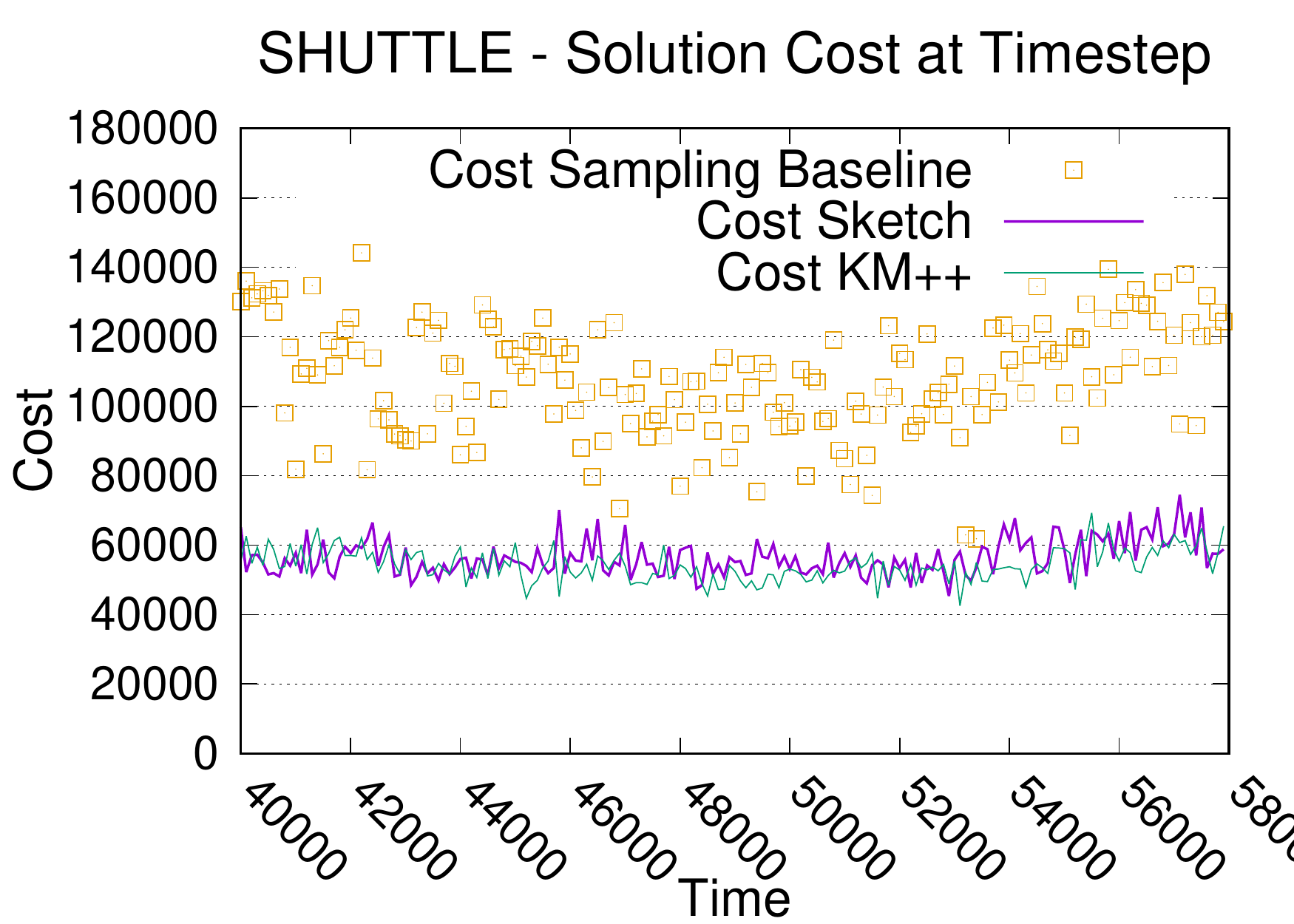}\label{fig:evolution-shuttle-ext}}
\subfigure[SKINTYPE]{\includegraphics[width=0.3\textwidth,keepaspectratio]{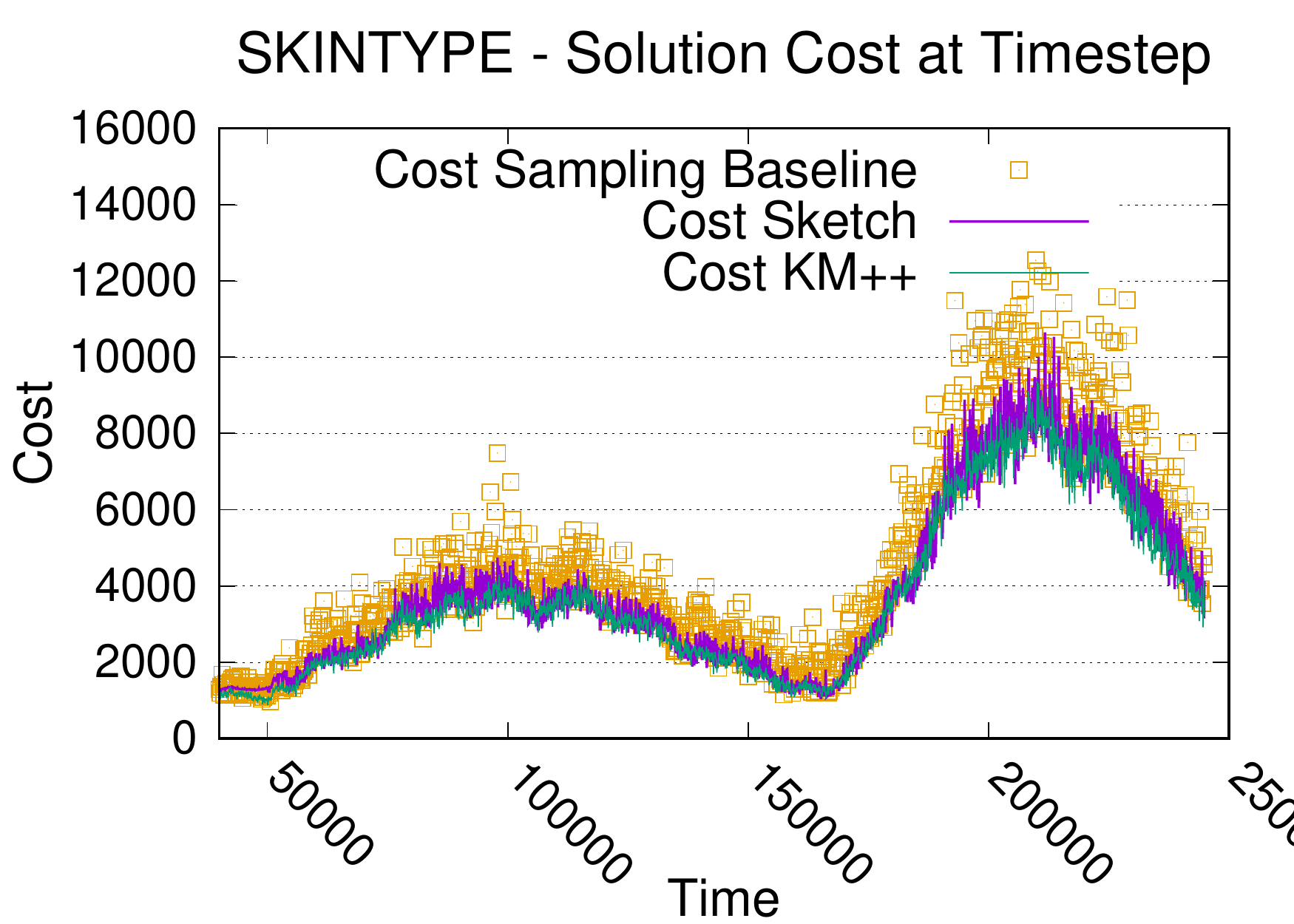}\label{fig:evolution-skin-ext}}
\caption{Cost of the solution obtained by our algorithm (Sketch) and the two baselines for $k=20$, $|W|=40{,}000$ and $\delta=0.2$. Notice that our algorithm's cost is close to that of the off-line algorithm and significantly better than the random sampling baseline}
\label{fig:evolution-ext}
\end{center}
\vspace{-0.15in}
\end{figure}

Observe that our algorithm closely tracks that of the optimum, even as the cost fluctuates up and down. Notice how the random  sampling baseline presents a particularly high approximation factor in COVERTYPE (Fig.~\ref{fig:evolution-covtype-ext}) during timesteps when there is a shift in the cost magnitude; this confirms the importance of sliding window algorithms to detect correctly (and re-cluster appropriately) shifts in data promptly. Similar results sre observed in all datasets Figure~\ref{fig:evolution-ext}.
To capture this analytically, we show the median relative error in Table~\ref{table:cost-over-baseline}. While our algorithm has a median error within 10\% of KM++, the  median error of the sampling baseline is 30-40\% higher than KM++. Notice how the size of the sliding window does not affect substantially the experimental results.

\subsection{Additional experiments on update time and space tradeoff}

\begin{table}
\centering
\begin{tabular}{ c | c | c | c |}
$W$ & $k$ & \textbf{Space}  & \textbf{Time}  \\
\hline
\multirow{3}{*}{20000}&10 &$7.2 \%$ $ (0.75\%)$&$1.2\%$ $ (0.68\%)$\\
&20&$13.72\% $ $ (2.06\%)$&$2.6\%$ $ (1.5\%)$\\
&40&$21.3\% $ $(3.4\%)$&$3.5\%$ $ (2.8\%)$\\
\hhline{~---}
\multirow{3}{*}{40000}&10&$3.5\%$ $ (0.39\%)$&$0.45\%$ $ (0.29\%)$\\
&20&$6.5\% $ $ (0.87\%)$&$0.93\% $ $(0.63\%)$\\
&40&$11.3\%$ $ (1.74\%)$&$1.58\%$ $ (1.23\%)$\\
\end{tabular}
\caption{ Max percentage of sliding window stored (Space) and median percentage of time (Time) vs. one run of $k$-means++ (stdddev in parantesis).}
\label{table:time-space-over-baseline2}
\vspace{-0.15in}
\end{table}

Additional experiments are presented in Table~\ref{table:time-space-over-baseline2} where we observe substantial savings w.r.t. the naive method. Notice how the savings increase when the size $w$ grows.

\subsection{Recovering ground-truth clusters.} 
In this section we evaluate the accuracy of the clusters produced by our algorithm on a dataset with ground-truth clusters. Note that the ground truth is for all the all set so in our setting we use the cluster assignments of the original ground truth restricted to the particular sliding window. We compared the offline KM++ baseline (KM++), our algorithm (Sketch), and the same-space (Sampling) baseline. We do not evaluate the SODA16 baseline because it is dominated by KM++ in space use and solution quality. We use the well known V-Measure accuracy definition for clustering~\cite{rosenberg2007v}. 

We report in Table~\ref{table:ground-truth-extended} the results for our empirical evaluation of the accuracy of our algorithm on identifying the ground truth clusters in synthetic datasets. 

We observe that our algorithm always dominates the Sampling baseline for all parameters and datasets and it has results comparable to the gold standard offline algorithm while using significantly less space and running time. 

\begin{table}
\centering
\begin{tabular}{llrrr}
\toprule
   &           &  \textbf{Off-line} &  \textbf{Sketch}&  \textbf{Sampling}\\
\textbf{Dataset} & \textbf{W} &                     &                     &                             \\
\midrule
s1 & 2000 &            0.958 (0.016) &      0.964  (0.019)&                    0.927 (0.026)\\
   & 3000 & 0.965 (0.015) &            0.969 (0.015)&                    0.926 (0.025)\\
   & 4000 & 0.969 (0.012) &             0.969 (0.017)&                    0.933 (0.030)\\
\hline
s2 & 2000 & 0.901 (0.018)&            0.903 (0.015)&                    0.856 (0.026)\\
   & 3000 & 0.889 (0.014)&            0.903 (0.014)&                    0.860 (0.028)\\
   & 4000 & 0.901 (0.014)&            0.903 (0.021)&                    0.851 (0.025)\\
\hline
s3 & 2000 & 0.741 (0.018)&            0.748 (0.013)&                    0.717 (0.019)\\
   & 3000 & 0.740 (0.017)&            0.745 (0.011)&                    0.717 (0.015)\\
   & 4000 & 0.735 (0.018)&            0.746 (0.017)&                    0.712 (0.013)\\
\hline
s4 & 2000 & 0.678 (0.015)&            0.677 (0.017)&                    0.665 (0.019)\\
   & 3000 & 0.678 (0.017)&            0.674 (0.015)&                    0.659 (0.019)\\
   & 4000 & 0.676 (0.009)&            0.675 (0.013)&                    0.666 (0.013)\\
\bottomrule
\end{tabular}
\caption{Average V-Measure of the cluster using the $k$-means++, sampling baselines and, our algorithm as compared to the ground truth clusters in synthetic datasets (stddev in paranthesis). We use $\delta=0.2$.}
\vspace{-0.2in}
\label{table:ground-truth-extended}
\end{table}

\subsection{Additional empirical results on comparisons with ths  baselines}
 
\paragraph{Max cost ratios vs gold standard}

\begin{table}
\small
\centering
\begin{tabular}{ c | c | c | c | c |}
\textbf{Dataset} & $W$ & \textbf{Sketch} & \textbf{Sampling}  & \textbf{Reduction}  \\
\hline
\multirow{3}{*}{COVER}&10000&1.64&4.09&$-60.0\%$ \\
&20000&1.49&4.05&$-63.1\%$\\
&40000&1.62&3.79&$-57.2\%$\\
\hline
\multirow{3}{*}{SHUTTLE}&10000&2.03&2.53&$-20.0\%$\\
&20000&1.69&1.80&$-6.3\%$\\
&40000&1.45&1.75&$-17.00\%$\\
\hline
\multirow{3}{*}{SKIN}&10000&2.28&4.16&$-45.3\%$\\
&20000&1.19&4.13&$-54.2\%$\\
&40000&1.76&3.27&$-46.3\%$\\
\end{tabular}
\caption{Maximum cost ratio of our algorithm and the sampling baseline over the $k$-means++ gold standard for $k=10$, $\delta=0.2$}
\label{table:cost-over-baseline-max}
\end{table}

Our improvement over this baseline is even larger if we look at max error (Table~\ref{table:cost-over-baseline-max}) over the entire stream.  Notice that, in these experiments, our algorithm always has less than $2.2$ cost factor, while the sampling baseline can return costs that are $>4$x higher than the gold standard, thus resulting in a reduction of up to $-60\%$ in cost.  Note also that the empirical approximation ratios are much better than what the worst case theoretical analysis pessimistically predicts.

\begin{figure}
\begin{center}
\includegraphics[width=0.6\textwidth,keepaspectratio]{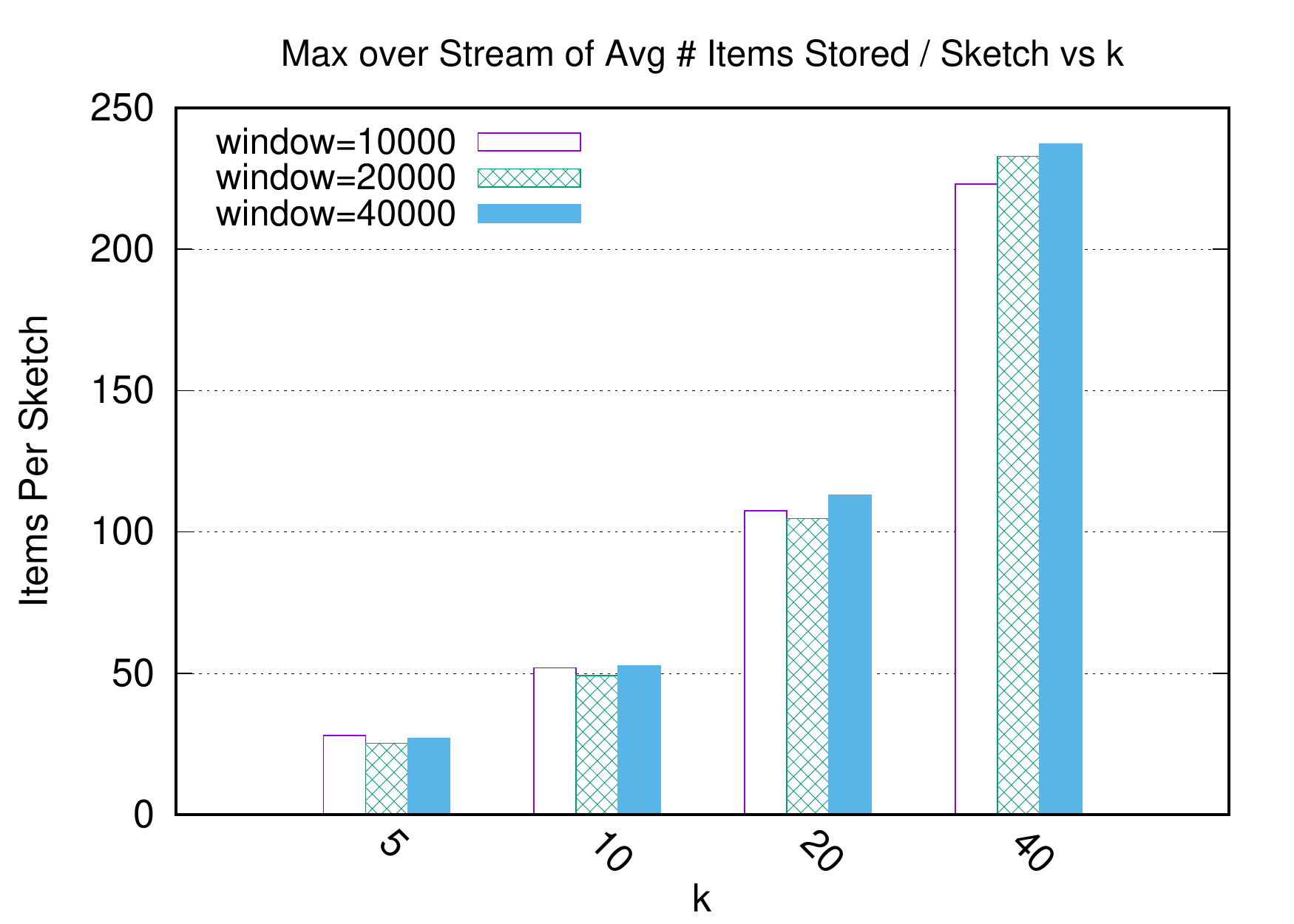}
\caption{Number of items per sketch stored by our algorithm.}
\label{fig:space-per-sketch}
\end{center}
\end{figure}

We report in Table~\ref{table:time-space-over-baseline-extended} additional results. 
\begin{table}
%\small
\centering
\begin{tabular}{ c | c | c | c | c |}
\textbf{Dataset} & $W$ & $k$ & \textbf{Space}  & \textbf{Time}  \\
\hline
COVER & 10000 & 5  &              7.24 &                                1.57 \\
     &       & 10 &             14.43 &                                3.12 \\
     &       & 20 &             28.33 &                                6.56 \\
     &       & 40 &             37.22 &                                6.33 \\
     & 20000 & 5  &              3.66 &                                0.62 \\
     &       & 10 &              7.28 &                                1.22 \\
     &       & 20 &             13.72 &                                2.62 \\
     &       & 40 &             21.08 &                                3.51 \\
     & 40000 & 5  &              1.83 &                                0.24 \\
     &       & 10 &              3.51 &                                0.46 \\
     &       & 20 &              6.53 &                                0.93 \\
     &       & 40 &             11.34 &                                1.59 \\
SHUTTLE & 10000 & 5  &              3.86 &                                0.57 \\
     &       & 10 &              5.77 &                                0.88 \\
     &       & 20 &              8.23 &                                1.60 \\
     &       & 40 &              4.81 &                                0.58 \\
     & 20000 & 5  &              1.60 &                                0.18 \\
     &       & 10 &              1.72 &                                0.17 \\
     &       & 20 &              2.22 &                                0.20 \\
     &       & 40 &              2.22 &                                0.18 \\
     & 40000 & 5  &              0.72 &                                0.09 \\
     &       & 10 &              0.69 &                                0.05 \\
     &       & 20 &              0.90 &                                0.07 \\
     &       & 40 &              0.94 &                                0.05 \\
SKIN & 10000 & 5  &              7.08 &                                1.61 \\
     &       & 10 &              9.81 &                                1.18 \\
     &       & 20 &              7.09 &                                0.43 \\
     &       & 40 &              9.14 &                                0.29 \\
     & 20000 & 5  &              3.26 &                                0.62 \\
     &       & 10 &              4.15 &                                0.51 \\
     &       & 20 &              3.38 &                                0.22 \\
     &       & 40 &              4.64 &                                0.16 \\
     & 40000 & 5  &              1.67 &                                0.26 \\
     &       & 10 &              2.46 &                                0.28 \\
     &       & 20 &              1.85 &                                0.12 \\
     &       & 40 &              2.10 &                                0.08 \\
\end{tabular}
\caption{Max percentage of the items stored w.r.t. the sliding window (Space) and median percentage of distance evaluation vs. one run of $k$-means++ for $\delta=0.2$}
\label{table:time-space-over-baseline-extended}
\end{table}

\paragraph{Comparison of costs with baseline}
We now analyze quantitatively the cost of the solutions obtained by our algorithm and the other baselines. Table~\ref{table:cost-over-baseline} reports the median ratio (over all timesteps evaluated) between the cost obtained by our algorithm and the $k$-means++ gold standard (column Sketch) as well as the median ratio of the sampling baseline over the gold standard (Sampling) for $k=10$, $\delta=0.2$, and various $|W|$'s. Notice how our algorithm has a median approximation factor (w.r.t. the gold standard) of $5-13\%$ higher cost and a reduction of cost w.r.t. the sampling baseline of up to $-25\%$. An approximation factor of $<1.2$ is significantly better than the pessimistic worst-case theoretical analysis and shows that our algorithm is highly precise in practice.

\paragraph{Fraction of window stored}

\begin{figure}
\begin{center}
\includegraphics[width=0.5\textwidth,keepaspectratio]{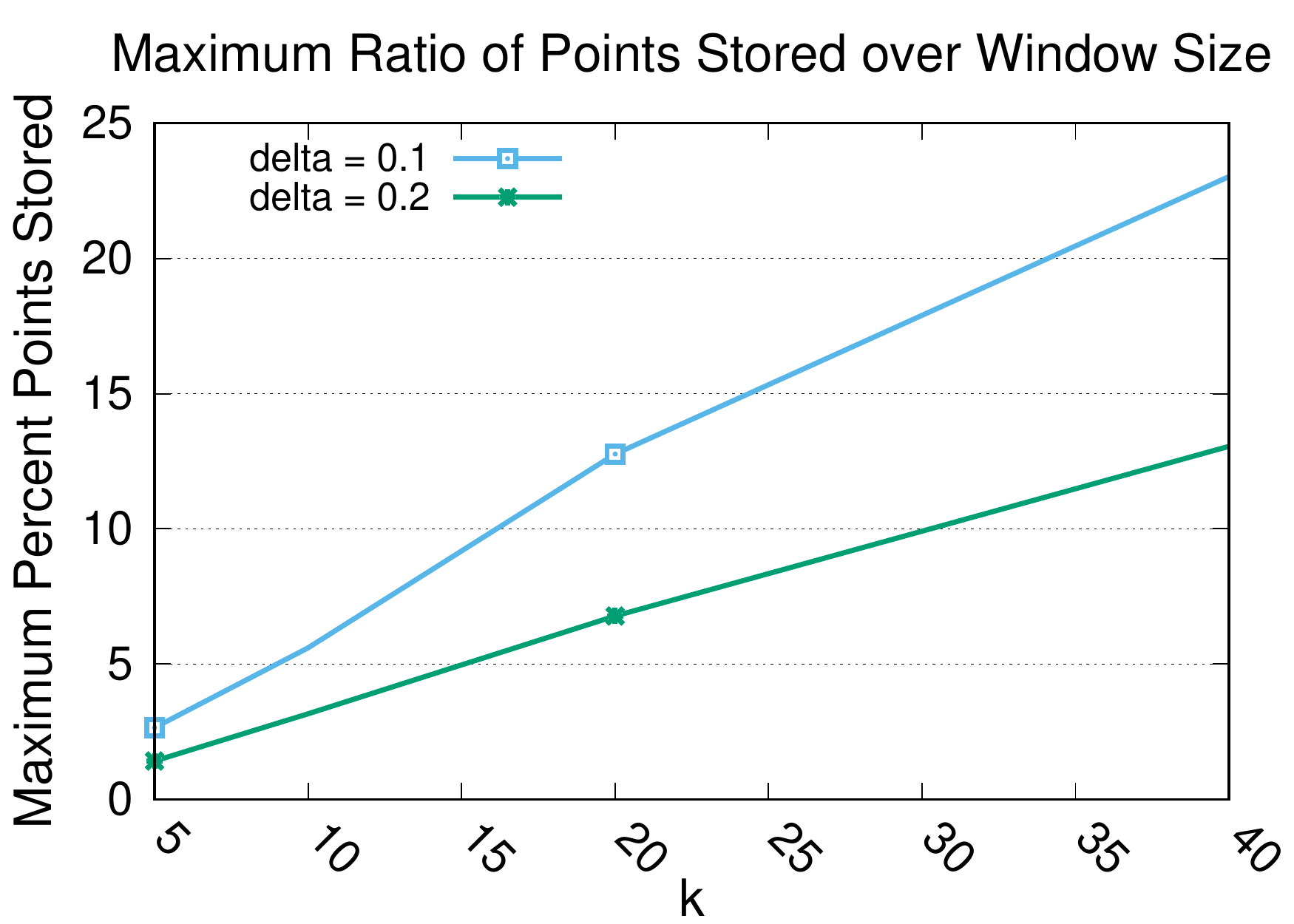}
\caption{Space usage of our algorithm for different values of $k$ and $\delta$ and $W=40{,}000$ in COVERTYPE. Space usage is reported as percent of he window points stored .}
\label{fig:space-time-items}
\end{center}
\end{figure}

 Figure~\ref{fig:space-time-items} we show the percent of the sliding window data points stored as a function of the desired number of clusters, $k$, for the COVERTYPE dataset and $W =40{,}000$ (other results are similar). Observe that we always store a small fraction of the dataset, and that the bounds grow linearly with $k$, as suggested by the theory. In %Figure~\ref{fig:space-time-dist}, we show the median update time as compared to a single run of $k$-means++ over the sliding window. Again, the running time grows linearly in $\frac{1}{\delta}$ and roughly quadratically in $k$ (as the number of centers stored grows linearly in $k$ and the running time of reclustering with $k$-means++ grows linearly in $k$). Again, we notice that in this range of parameters, the algorithm is much faster than the naive approach, requiring a small fraction of the time. 
Finally, In Figure~\ref{fig:space-per-sketch} we show the average number of points stored per individual sketch. Notice how the number of points stored is mostly unaffected by $w$, thus confirming that our algorithm can scale to large window sizes.

\subsection{Experiments using a different $p$}
All the experiments we have reported are for k-means ($p=2$). We observe that our implementation easily adapts to any $p$ by just changing the underlying $k$-clustering solver. For sake of completeness, we also ran some experiments with the $k$-median objective, $p=1$. All the results we observe are in line with that of the k-means case ($p=2$). For instance, replicating the clustering accuracy experiment using $p=1$ (See Table~\ref{table:ground-truth-extended} for the results for $p=2$) we obtain the following V-Measure accuracy for  the offline algorithm, our stream algorithm and the sampling baseline (same space), $81.8\%$, $79.9\%$, $78.5\%$, respectively. This confirming the same trend observed for k-means. We observe the same trend in terms of speedups and space savings w.r.t. the gold-standard baseline.

%% file: app-assumption.tex
\section{Appendix: Relaxing the assumptions on the input}\label{app:assumption}

In the paper we made the following assumptions for the sake of exposition: 
\begin{itemize}
\item the distances between points are normalized to lie between $1$ and $\Delta$,
\item  the points are distinct, and 
\item we have access to a lower bound $m$ and upper bound $M$ of the cost of the optimal solution in any sliding window.
\end{itemize}
We can relax these assumptions as follows.  

\subsection{Relaxing the distinct points' assumption.}
Notice that we made the assumption that the points in $X$ are distinct only for sake of simplifying the notation of the mapping function from points to centers (we want to define a function from a point). In general, points can be repeated in the stream and without any loss of generality one can define a function (implicitly) on the points assuming their arrival time make them distinct entities (note that we do not need to store the mapping which is not used by the algorithm, this is only for sake of analysis).

\subsection{Relaxing the assumption having distances in  $[1,\Delta]$ in Meyerson algorithm} 
Notice that in the Meyerson algorithm (Algorithm~\ref{alg:meyerson})  we assume that the min distance of two non-coincident points is $1$ and we assume to know the maximum distance $\Delta$. The latter value is used to sample points with probability depending on $\log(\Delta)$. In reality we do not need such assumptions. Notice that the $\log(\Delta)$ factor is needed to perform a union bound over the number of exponentially growing balls around a center that starts at distance given by the average cost of a point in a cluster. It is possible to observe that the maximum distance from a point in a cluster $C$ to the center, is a most $|C|$ times the average distance of points to the center in the cluster. So if we assume that the algorithm is run on a stream of size at most $O(w)$ we can use $\log(w)$ instead of $\log(\Delta)$ in the algorithm, and ignore the assumption on the min distance of points. In the Section~\ref{app:short-stream} we show that we never need to compute the sketches on streams of length larger than $2w$ so this we can replace $\log(\Delta)$ with $\log(2w)$ in Algorithm~\ref{alg:meyerson}.

\subsection{Estimating the min cost $m$}
The final assumption is knowing a lower bound on the minimum cost $m$ and an upper bound on cost, $M$ for defining the grid containing the guesses of the optimum. To have a lower bound on $m$, we maintain the most recent $k+1$ distinct points at all times in a set $S$. 
Whenever a new point $x$ arrives, we check to see if we can replace one of the $k+1$ points in $S$ with $x$. For instance if $x \notin S$, we add $x$ to $S$ and remove the most stale (oldest) point in $S$. If an identical point $y = x$ already exists in set $S$, we just replace $y$ with this new copy $x$. This way, we always keep the most recent distinct $k+1$ points in $S$. If some of the points in $S$ lie outside the current active window $W$, the total number of distinct points in the active window does not exceed $k$. In this case, the optimum clustering has cost zero and we have access to the $\leq k$ centers yielding this optimum clustering. Otherwise, all points in $S$ belong to the active window $W$. Therefore the minimum cost of clustering $S$ lower bounds the optimum clustering cost of the whole window $W$. Notice that the minimum cost of clustering $S$ is simply the minimum pairwise distance between points in $S$. This is a valid lower bound on $m$. 

As the sliding window shifts, the lower bound on $m$ might increase or decrease. Whenever, the lower bound increases, in algorithm $SimpleMeyerson(X_t)$, we need to discard some runs of $ComputeMeyerson(X_t, L^p)$. These are instances where $L^p$ is smaller than the new lower bound on $m$ and therefore not relevant anymore. Discarding those runs is easy to do. 

However, when the lower bound on $m$ decreases, we need to instantiate new runs of $ComputeMeyerson(X_t, L^p)$ for some smaller values of $L^p$. But these new runs are only valid for sliding windows that consist of only the current set of points in $S$. Older versions of $S$ were providing higher lower bounds on $m$ and therefore we do not need to do anything for those older sliding windows. So the only sliding window for which we need to run new instances of $ComputeMeyerson(X_t, L^p)$ is the sliding window starting at time when $S$ begins, which consists of only distinct points in $S$ so far. Notice that we can easily keep the multiplicities of the coincident points for each point in $S$. Therefore we have all the information we need to run the new instances of $ComputeMeyerson(X_t, L^p)$ for the new smaller values of $L^p$ simulating the algorithm on a arbitrary ordering of points (with multiplicities) from $S$.

\subsection{Estimating the max cost $M$}
Now we will talk about maintaining an upper bound $M$ on the optimum cost. The number of points in the window is $w$, and the optimum cost cannot be more than $w$ times the diameter of the points in the sliding window. Cohen-Addad et al. \cite{diameter} design an algorithm that approximates the diameter of points in a sliding window with a constant $3+\epsilon$ approximation guarantee and using $O(\log(\Delta)/\epsilon)$ memory. Therefore, we can always use this dynamically changing upper bound on $M$. Whenever our upper bound decreases, we need to discard some runs of $ComputeMeyerson(X_t, L^p)$ for values of $L^p$ that are greater than the new upper bound on $M$ which is easy to do. 

However, when the upper bound on $M$ increases, there are new values of $L^p$ for which we need to run $ComputeMeyerson(X_t, L^p)$. The challenge is to run these sketches retrospectively. We do not have access to the whole active window $W$. 
Let $L^{old}$ be the highest value of $L^p$ we had before we increased our upper bound on $M$.  
For any  new value of $L^p$ that exceeds our previous upper bound on $M$ (let's call this new value $L^{new}$), we initialize the sketch with the sketch we have maintained so far for $L^{old}$. In other words, we use the closest value of $L^p$ to initialize all new sketches. 

In the sampling procedure, $L^p$ is the denominator in the sampling probability. So reusing the sketch with lower $L^p$ means that the previous points have been over-sampled. The main properties of the sketch proved in Lemma~\ref{lem:meyerson} are about upper bounding the cost of Meyerson sketch and its size. The cost is upper bounded in the same way and the proof is intact as the sampling probabilities only increased. 

The upper bound on the size can be adapted as follows. We divide the stream into chunks using the times we increased our upper bound on $M$ as pivot points. Notice that we do not need to update $L^p$ each time $M$ changes but only when $M$ increases by a constant  factor. The sketch $L^{new}$ consists of the points we sample in each of these chunks. So we upper bound the size of the sketch in each chunk. The last chunk is similar to any other Meyerson sketch with a set of points being pre-selected in previous chunks. So the proof of Lemma~\ref{lem:meyerson} works here as well. For any other chunk $S'$, we are running the sketch with the maximum $L^p$ at the time the points in $S'$ are arriving. So this $L^p$ is at least  a constant fraction of the optimum cost of $S'$. This suffices for the argument to upper bound the size in Lemma~\ref{lem:meyerson}. To summarize, we have a new overall upper bound on the size of the sketch with an extra $\log(\Delta)$ as the number of chunks is related to logarithm of how much diameter grows (as we start a new chunks only when $M$ increases by a constant factor). 

\subsection{Ensuring that the each sketch is computed on at most $O(w)$ points}
\label{app:short-stream}
Finally, we now show that we never need to compute any streaming algorithm on a too long stream. This can be obtained by at most doubling the time and memory requirement of the algorithm with the following technique. We keep at most $2$ independent instances of our sliding window algorithm (Algorithm~\ref{alg:algo-min}) the following way. Every $w$ steps we start a new algorithm instance that receives points from this moment on. When an algorithm instance has consumed $2w$ points that algorithm is discarded, meanwhile we update the (at most $2$ active algorithms) with each point received. Any time a solution is needed we use the active algorithm that has started earlier. It is easy to see that for such algorithm all the properties needed by the theorem~\ref{th:algo-main-th} are true as the algorithm has consumed the entire active window, so the solution is good. Notice that this ensures that no sketch consumes more than $2w$ points. Also notice that the running time and the memory is doubled at most.